\documentclass[11pt]{article}
\usepackage{fullpage}
\newif\ifFull
\Fullfalse

\ifFull
\else
\fi

\usepackage{fullpage,xspace,amsmath,amsfonts,amssymb}
\usepackage{times,mathptmx}
\usepackage{graphicx,color,hyperref}
\usepackage{amsfonts,amsmath,color,float,graphicx,verbatim, times}
\usepackage{amsthm}

 \newtheorem{theorem}{Theorem}
 \theoremstyle{plain}

 \newtheorem{lemma}{Lemma}

 \numberwithin{equation}{section}
 
 \newcommand{\edit}[1]{{{#1}}}
  \newcommand{\editmm}[1]{{{#1}}}
\begin{document}
\title{Parallel Peeling Algorithms}
\author{Jiayang Jiang \thanks{\noindent Harvard University, School of Engineering and Applied Sciences. Supported by NSF grants CCF-0915922 and IIS-0964473.} \and Michael Mitzenmacher\thanks{\noindent Harvard University, School of Engineering and Applied Sciences. Supported in part by NSF grants CCF-0915922, IIS-0964473, and CNS-1228598.} \and Justin Thaler \thanks{\noindent Yahoo! Labs. The majority of this work was performed while the author was a graduate student at Harvard University, School of Engineering and Applied Sciences. Parts of this work were performed while the author was a Research Fellow at the Simons Institute for the Theory of Computing at UC Berkeley. Supported by an NSF Graduate Research Fellowship, NSF grant CCF-0915922, and a Simons Research Fellowship.}}
\date{}
\maketitle

\begin{abstract}

The analysis of several algorithms and data structures can be framed
as a {\em peeling process} on a random hypergraph: vertices with
degree less than $k$ are removed until there are no vertices of degree
less than $k$ left. The remaining hypergraph is known as the $k$-core.  In
this paper, we analyze parallel peeling processes, where in each
round, all vertices of degree less than $k$ are removed.  It is known
that, below a specific edge density threshold, the $k$-core is empty
with high probability.
We show that, with high probability, below
this threshold, only $\frac{1}{\log((k-1)(r-1))}\log \log n + O(1)$ rounds of peeling are needed to
obtain the empty $k$-core for $r$-uniform hypergraphs; this bound is tight up to an additive constant. 
Interestingly, we show that above this
threshold, $\Omega(\log n)$ rounds of peeling are
required to find the non-empty $k$-core.  Since most algorithms and
data structures aim to peel to an empty $k$-core, this asymmetry appears
fortunate.  We verify the theoretical results both with 
simulation and with a parallel implementation using graphics processing
units (GPUs). Our implementation provides insights into how to structure parallel 
peeling algorithms for efficiency in practice.  
\end{abstract}
\section{Introduction}
Consider the following {\em peeling process}: starting with a random
hypergraph, vertices with degree less than $k$ are repeatedly removed,
together with their incident edges. (We use edges instead of hyperedges throughout the paper,
as the context is clear.)  This yields what is called the
$k$-core of the hypergraph, which is the maximal subgraph where each
vertex has degree at least $k$.  It is known that the $k$-core is uniquely defined and
does not depend on the order vertices are removed.  
The greedy peeling process produces sequential algorithms with
very fast running times, generally linear in the size of the graph.
Because of its simplicity and efficiency, peeling-based approaches appear especially useful for problems involving large data sets.
Indeed, this process, and variations on it, have found
applications in low-density parity-check codes \cite{LMSS,MV},
hash-based sketches \cite{Bloomier,GM}, satisfiability of random boolean
formulae \cite{BFU,Molloy}, and cuckoo hashing \cite{PR}.  Frequently,
the question in these settings is whether or not the $k$-core is
empty.  As we discuss further below, it is known that below a specific
edge density threshold $c^*_{k, r}$, the $k$-core is empty with high probability.
This asymptotic result in fact accurately predicts practical
performance quite well.

  In
this paper, we focus on expanding the applicability of peeling processes by
examining the use of parallelism in conjunction with peeling.  Peeling
seems particularly amenable to parallel processing via the following simple round-based algorithm:
in each round, all vertices of degree less than $k$ and their adjacent edges are removed in
parallel from the graph.  The
major question we study is: how many rounds are necessary before peeling is complete?

We show that, with high probability, when the edge density is a constant strictly below
the threshold $c_{k, r}^*$, only $\frac{1}{\log ((k-1)(r-1))}\log \log n + O(1)$ rounds of
peeling are needed for $r$-uniform hypergraphs.
(The hidden constant in the $O(1)$ term depends on the size of the ``gap'' between the edge density and the threshold density. We more precisely characterize this dependence later.) 
Specifically, we show that the fraction of vertices that remain
in each round decreases doubly exponentially, in a manner similar in spirit to existing analyses of
``balanced allocations'' load-balancing problems \cite{ABKU,Mitzenmacher}. 
Interestingly, we show in contrast that at edge densities above the threshold,
with
high probability 
$\Omega(\log n)$ rounds of peeling are required to
find the non-empty $k$-core.  Since most algorithms and data
structures that use peeling aim for an empty $k$-core, the fact that
empty $k$-cores are faster to find in parallel than non-empty ones appears particularly fortuitous. 


We then consider some of the details in implementation, focusing on
the algorithmic example of Invertible Bloom Lookup Tables (IBLTs)
\cite{GM}.  An IBLT stores a set of keys, with each key being hashed
into $r$ cells in a table, and all keys in a cell XORed together.
The IBLT defines a random hypergraph, where keys correspond to edges,
and cells to vertices.  As we describe later, recovering the set of
keys from the IBLT corresponds to peeling on the associated
hypergraph.  Applications of IBLTs are further discussed in \cite{GM};
they can be used, for example, for sparse recovery \cite{GM},
simple low-density parity-check codes \cite{MV}, and 
efficient set reconciliation across communication links \cite{EGUV}.
Our implementation demonstrates that our parallel peeling
algorithm yields concrete speedups, and provides insights into how to structure parallel 
peeling algorithms for efficiency in practice.  


\edit{Our results are closely related
to work of Achlioptas and Molloy \cite{achlio}. 
With different motivations than our own, they
show that at most $O(\log n)$ rounds of peeling are needed to find the (possibly non-empty)
$k$-core
both above and below the threshold edge density
$c_{k, r}^*$. Our $O(\log \log n)$ upper bound below the
threshold is an exponential improvement on their $O(\log n)$ bound,
while our $\Omega(\log n)$ lower bound above the threshold
demonstrates the tightness of their upper bound in this regime.
Perhaps surprisingly, we cannot find other analyses of parallel peeling
in the literature, although early work by Karp, Luby, and Meyer auf
der Heide on PRAM simulation uses an algorithm similar to peeling to
obtain $O(\log \log n)$ bounds for load balancing \cite{KLM}, and we use other
load balancing arguments \cite{ABKU,vocking} for inspiration. 
We also rely heavily on the framework established by Molloy \cite{Molloy}
for analyzing the $k$-core of random hypergraphs.}

Subsequent to our work, Gao \cite{gao} has provided an
alternative proof of an $O(\log \log n)$
upper bound on the number of rounds required 
to peel to an empty core when the edge density is below the 
threshold
$c_{k, r}^*$.
Her proof,
short and elegant, obtains a leading constant 
of $\frac{1}{\log\left(k(r-1)/r\right)}$, larger than the constant 
$\frac{1}{\log ((k-1)(r-1))}$ obtained through our more detailed analysis. 

\medskip
\noindent \textbf{Paper Outline.}
Section \ref{sec:belowthresh} characterizes the round complexity
of the peeling process when the edge density is a constant strictly below
the threshold $c_{k, r}^*$, showing that the number of rounds
required is
 $\frac{1}{\log((k-1)(r-1))}\log \log n + O(1)$.
 Section \ref{sec:above} shows that when the edge density
 is a constant strictly above the threshold $c_{k, r}^*$, the number of rounds
 required is $\Omega(\log n)$. Section \ref{sec:experiment}
 presents simulation results demonstrating that our 
 theoretical analysis closely matches the empirical evolution of 
 the peeling process.
  Section \ref{sec:GPU} describes our GPU-based IBLT implementation. 
\edit{Our IBLT implementation
must deal with a 
fundamental issue that is inherent to \emph{any}
 implementation of a parallel peeling algorithm,
 regardless of the application domain:
 the need to avoid peeling the same item  multiple times.}
 Consequently, 
the peeling process used in our IBLT implementation differs slightly from 
 the one analyzed in Sections \ref{sec:belowthresh} and \ref{sec:above}. 
 \edit{In Appendix \ref{sec:peelsubtable}
we formally analyze this variant of the parallel peeling process, demonstrating that it terminates significantly faster
than might be expected.}

As discussed above, the hidden constant in the additive $O(1)$ term in the upper bound of Section \ref{sec:belowthresh} depends on the distance between the edge density
and the threshold density $c^*_{k, r}$;  we refer to this distance as $\nu$. Section \ref{sec:distance} extends the analysis
 of Section \ref{sec:belowthresh} to precisely characterize 
this dependence, demonstrating 
that there is an additive $\Theta(1/\sqrt{\nu})$ term in the number of rounds
required.
Section \ref{sec:conclusion} concludes.

\section{Preliminaries}
For constants \edit{$r \geq 2$} and $c$, let $G_{n, cn}^r$ denote a random
hypergraph\footnote{When $r=2$ we have a graph, but we may use hypergraph when speaking generally.} with $n$ vertices and $cn$ edges, where each
edge consists of $r$ distinct vertices. Such hypergraphs are called \emph{r-uniform}, and we refer to $c$ as the \emph{edge density}
of $G_{n, cn}^r$.  Previous analyses of
random hypergraphs have determined the threshold
values $c_{k,r}^*$ such that when $c < c_{k,r}^*$, the $k$-core is empty with
probability $1-o(1)$, and when $c > c_{k,r}^*$, the $k$-core is non-empty
with probability $1-o(1)$.  \edit{Here and throughout this paper, $k,r \geq 2$, but
the special (and already well understood) case of $k = r =2$ is excluded from consideration.}
From \cite{Molloy}, the formula for $c_{k,r}^*$ is given
by
\begin{align}
\label{eq:cstar}
c_{k,r}^* = \min_{x > 0} \frac{x}{r(1-e^{-x}\sum_{j=0}^{k-2}\frac{x^j}{j!})^{r-1}}.
\end{align}
For example, we find that $c_{2,3}^* \approx 0.818$, $c_{2,4}^* \approx 0.772$ and $c_{3,3}^* \approx 1.553$.

\section{Below the Threshold}
\label{sec:belowthresh}

In this section, we characterize the number of rounds required by
the peeling process when the edge density $c$ is a constant strictly below the threshold 
density $c^*_{k, r}$.  Recall that 
this peeling process repeatedly removes vertices with degree less than $k$,
together with their incident edges. 
We prove the following theorem.

\begin{theorem}
\label{thm:loglog}
\edit{Let $k, r \geq 2$ with $k+r \geq 5$}, and let $c$ be a constant. With probability $1-o(1)$, the parallel peeling process for the $k$-core in a random hypergraph $G_{n, cn}^r$ with
edge density $c$ and
$r$-ary edges terminates after $\frac{1}{\log((k-1)(r-1))}\log \log n + O(1)$ rounds when $c < c^*_{k, r}$.
\end{theorem}

\editmm{
Theorem \ref{thm:loglog} is tight up to an additive constant. 
\begin{theorem}
\label{thm:logloglb}
\edit{Let $k, r \geq 2$ with $k+r \geq 5$}, and let $c$ be a constant. With probability $1-o(1)$, the parallel peeling process for the $k$-core in a random hypergraph $G_{n, cn}^r$ with
edge density $c$ and
$r$-ary edges requires $\frac{1}{\log((k-1)(r-1))}\log \log n - O(1)$ rounds to terminate when $c < c^*_{k, r}$.
\end{theorem}
}


In proving Theorems \ref{thm:loglog} and \ref{thm:logloglb}, we begin in Section \ref{sec:high}
with a high-level overview of our argument, before presenting
full details of the proof in Section \ref{sec:compl}.
\subsection{The High-Level Argument}
\label{sec:high}



The neighborhood of a node $v$ in a random $r$-uniform hypergraph can be accurately modeled as a branching
process, with a random number of edges adjacent to this vertex, and
similarly a random number of edges adjacent to each of those vertices,
and so on.  For intuition, we assume this branching process yields a
tree, and further that the number of adjacent edges is distributed
according to a discrete Poisson distribution with mean $rc$.  These
assumptions are sufficiently accurate for our analysis, as we later
prove.  (This approach is standard;  see e.g. \cite{DGMMPR,Molloy} for similar arguments.)

The intuition for the main result comes from considering the (tree)
neighborhood of $v$, and applying the following algorithm: for $1 \leq i
\leq t-1$, in round $i$, look at all the vertices at distance $t-i$
and delete a vertex if it has fewer than $k-1$ child edges. Finally,
in round $t$, $v$ is deleted if it has degree less than $k$. 
Vertex $v$ survives after $t$ rounds of peeling if and only
if it survives after $t$ rounds of this algorithm.

In what follows, we denote the probability that $v$ survives after $t$ rounds in
this model by $\lambda_t$, and the probability a vertex $u$
at distance $t-i$ from $v$ survives $i$ rounds by $\rho_i$.  

Here $\rho_0 = 1$. In this idealized setting, the following relationships hold:
 \vspace{-3mm}
\begin{align*}
\rho_i  &= \Pr(\text{Poisson}(\rho_{i-1}^{r-1}rc) \geq k-1),
\end{align*}
and similarly
 \vspace{-2mm}
\begin{align}
\label{eq:lambdai}
\lambda_i  &= \Pr(\text{Poisson}(\rho_{i-1}^{r-1}rc) \geq k).
\end{align}

The recursion for $\rho_i$ arises as follows:  each node $u$ has a Poisson distributed number of descendant edges with mean $rc$, 
and each edge has $r-1$ additional vertices that each survive $i-1$ rounds with probability $\rho_{i-1}$.  By the splitting
property of Poisson distributions \cite[Chapter 5]{MU}, the number of surviving descendant edges of $u$ is Poisson distributed
with mean $\rho_{i-1}^{r-1}rc$, and this must be at least $k-1$ for $u$ to itself survive the $i$th round.  

We use $\beta_i$ to represent the expected number of surviving descendant edges after $i-1$ rounds:
\begin{align*}
\beta_i &= \rho_{i-1}^{r-1}rc.
\end{align*}

Then,
\begin{align}
\rho_i &= 1 - e^{-\beta_i}\sum_{j=0}^{k-2}\frac{{\beta_i}^j}{j!}, \\
\label{eq:lambdairecursion}
\lambda_i &= 1 - e^{-\beta_i}\sum_{j=0}^{k-1}\frac{{\beta_i}^j}{j!}, \\
\label{eq:betairecursion}
\beta_{i+1} &= \bigg[1 - e^{-\beta_i}\sum_{j=0}^{k-2}\frac{{\beta_i}^j}{j!}\bigg]^{r-1}rc.
\end{align}

When $c < c_k^*$, which is the setting where we know the core becomes empty, we have $\lim_{t \to \infty} \rho_t = 0$, so  $\lim_{t \to \infty} \beta_t = 0$. Thus, for any constant $\tau > 0$, we can choose a constant $I$ such that $\beta_I \leq \tau$.

For any $x > 0$ and $k \geq 2$, by basic calculus, we have 
\begin{equation}
\label{eq:calculus}1 - e^{-x}\sum_{j=0}^{k-2}\frac{x^j}{j!} \leq \frac{x^{k-1}}{(k-1)!}. \end{equation}
Applying this bound to $\beta_{I+1}$ gives
\ifFull
\begin{align*}
\beta_{I+1} &\leq \bigg[\frac{\beta_I^{k-1}}{(k-1)!}\bigg]^{r-1}rc \\
&\leq \beta_I^{(k-1)(r-1)}\frac{rc}{[(k-1)!]^{r-1}}.
\end{align*}
\else
$$\beta_{I+1} \leq \bigg[\frac{\beta_I^{k-1}}{(k-1)!}\bigg]^{r-1}rc \leq \beta_I^{(k-1)(r-1)}\frac{rc}{[(k-1)!]^{r-1}}.$$
\fi
Using induction, we can show that
\begin{align*}
\beta_{I+t} \leq \beta_I^{[(k-1)(r-1)]^t}\bigg[\frac{rc}{[(k-1)!]^{r-1}}\bigg]^{\frac{[(k-1)(r-1)]^t-1}{(k-1)(r-1)-1}}.
\end{align*}
If $\frac{rc}{[(k-1)!]^{r-1}} \geq 1$, we can apply the upper bound 
$$\beta_{I+t} \leq \big[\tau\big(\frac{rc}{[(k-1)!]^{r-1}}\big)^{\frac{1}{(k-1)(r-1)-1}}\big]^{[(k-1)(r-1)]^t},$$ and if $\frac{rc}{[(k-1)!]^{r-1}} < 1$, then $\beta_{I+t} \leq \tau^{[(k-1)(r-1)]^t}$. Setting $$\tau' = \max\Big(\tau\big(\frac{rc}{[(k-1)!]^{r-1}}\big)^{\frac{1}{(k-1)(r-1)-1}}, \tau\Big)$$ gives

\begin{align}
\label{eq:betaipt}
\beta_{I+t} \leq (\tau')^{[(k-1)(r-1)]^t}.
\end{align}
Pick $\tau$ such that $\tau' < 1$. By {Equations \eqref{eq:lambdairecursion}, \eqref{eq:calculus}, and \eqref{eq:betaipt}},
it holds that
\ifFull
\begin{align*}
\lambda_{I+t} &\leq \frac{\beta_{I+t}^k}{k!} \\
&\leq \frac{(\tau')^{k[(k-1)(r-1)]^t}}{k!}.
\end{align*}
\else
$$ \lambda_{I+t} \leq \frac{\beta_{I+t}^k}{k!} \leq \frac{(\tau')^{k[(k-1)(r-1)]^t}}{k!}.$$
\fi

Solving $\frac{(\tau')^{k[(k-1)(r-1)]^t}}{k!} < n^{-2}$ gives $t > \frac{1}{\log((k-1)(r-1))}\log \log n + O(1)$. This shows that it takes $t^* = \frac{1}{\log((k-1)(r-1))}\log \log n + O(1)$ rounds
for $\lambda_tn = o(1)$ in our idealized setting.  

%

\medskip
\editmm{
\noindent {\em Remark:} One can similarly show that with probability $1-o(1)$ 
termination requires at least $\frac{1}{\log((k-1)(r-1))}\log \log n - O(1)$ rounds 
for any constant $c < c^*_{k, r}$ when $k+r \geq 5$ as well in the idealized setting.  Starting from Equation
\eqref{eq:calculus}, we can show
\begin{equation*}
1 - e^{-x}\sum_{j=0}^{k-2}\frac{x^j}{j!} \geq \frac{x^{k-1}}{C(k-1)!}\end{equation*}
for some constant $C$ and sufficiently small $x > 0$.  
It then follows by similar arguments that
\begin{align*}
\beta_{I+t} \geq (\tau'')^{[(k-1)(r-1)]^t}
\end{align*}
for suitable constants $I$ and $\tau''$.
In particular, we can choose a $t$ that is $\frac{1}{\log((k-1)(r-1))}\log \log n - O(1)$, 
so that the number of vertices that remain to be peeled after $t$ rounds is stil
at least $n^{2/3}$ in expectation.  As we show later (cf. Section~\ref{sec:thm:loglogb}),
the fact that this expectation is large
implies that the number of surviving vertices after this many rounds is bigger than 0 with
probability $1-o(1)$, in both the idealized setting considered in this overview,
and in the actual random process corresponding to $G_{n, cn}^r$.}


\subsection{Completing the Argument}
\label{sec:compl}

\subsubsection{Preliminary Lemmas}
To formalize the argument outlined in Section \ref{sec:high}, we first note that instead of working
in the $G_{n, cn}^r$ model, we adopt the standard approach of having
each edge appear independently in the hypergraph with probability $q= cn/{n
  \choose r}$.  It can be shown easily that the result in this model (which we denote by $G^r_c$)
implies that the same result holds in the $G_{n, cn}^r$ model (see e.g. \cite{DGMMPR,KMW,Molloy}). Here, we sketch a simple version of this
standard argument for this setting.
\begin{lemma}
\label{lem:equivmodels}
Let $G^r_c$ be an $r$-uniform hypergraph on $n$ vertices in which each edge appears independently with probability $q= cn/{n \choose r}$.
Suppose that for all $c < c^*_{k,r}$, peeling succeeds on $G^r_c$ in $\frac{1}{\log((k-1)(r-1))}\log \log n + O(1)$ rounds with probability
$1-o(1)$. Then peeling similarly
succeeds on $G_{n, cn}^r$ in $\frac{1}{\log((k-1)(r-1))}\log \log n + O(1)$ rounds with probability $1-o(1)$  for all $c < c^*_{k,r}$.
\end{lemma}

\begin{proof} (Sketch)  Let $c'$ be a constant value (independent of $n$) with $c < c' < c^*_{k,r}$.   
With probability $1-o(1)$, parallel peeling will succeed for the hypergraph $G^r_{c'}$ in the appropriate number of rounds. Moreover, by standard Chernoff
bounds, $G^r_{c'}$ will have greater than $cn$ edges with probability $1-o(1)$.  Since the probability that the parallel peeling algorithm succeeds after
any number of rounds monotonically decreases with the
addition of random edges, it holds that the success probability is also $1-o(1)$ when the graph is chosen from 
$G_{n, cn}^r$.  \edit{(Formally, one would first condition on the number of edges chosen on the graph $G^r_{c'}$;  given the number of edges, the actual edges selected are random.  Hence we can couple the choice of the first $cn$ edges between the two graphs.)}
\end{proof}

We will also need the following lemma, which is essentially due to Voll \cite{voll}. We provide the proof for completeness.  (We have not aimed to optimize the constants.)
\begin{lemma}
\label{lem:polylog}
For any constants $c, r, c_1>0$, there is a constant $c_2>0$ such that with probability $1-1/n$, for all vertices $v$ in $G_c^r$, the neighborhood of distance $c_1 \log \log n$ around $v$ contains at most $\log^{c_2}n$ vertices.  
\end{lemma}

\begin{proof}
We follow the approach used in the dissertation of Voll \cite[Lemma 3.3.1]{voll}. 
Denote by $N_d$ the number of vertices at distance $d$ in the neighborhood of a root vertex $u$.
We prove inductively on $d$ that 
$$\Pr(N_d > (6cr^2)^d\log(1/\epsilon)) \leq d\epsilon$$
for $d$ up to $c_1 \log \log n$ and $\epsilon = 1/n^2$.  
The claim then follows by a union bound over all $n$ vertices $u$.

For convenience we assume $6cr \geq 1$;  
the argument is easily modified if this is not the case, instead proving $\Pr(N_d > r^d\log(1/\epsilon)) \leq d\epsilon$.
Recall that the number of edges adjacent to $u$ is dominated by a binomial random variable 
$B\left ( {n-1 \choose r-1}, q \right )$, which has mean $cr$.  The number of vertices adjacent
to $u$ via these edges is dominated by $r-1$ times the number of edges. 
When $d = 1$, we find that the number of neighboring
edges of the root, which we denote by $N'_0$, is at most $6cr \log(1/\epsilon)$ with probability bounded above by
\ifFull
\begin{eqnarray*}
{{n-1 \choose r-1} \choose 6cr \log(1/\epsilon) } q^{6cr \log{1/\epsilon}} & \leq & \left ( \frac{ecr}{6cr \log(1/\epsilon)} \right)^{6cr \log(1/\epsilon)} \\
& \leq & \epsilon.
\end{eqnarray*}
\else
$${{n-1 \choose r-1} \choose 6cr \log(1/\epsilon) } q^{6cr \log{1/\epsilon}} \leq \left ( \frac{ecr}{6cr \log(1/\epsilon)} \right)^{6cr \log(1/\epsilon)}
\leq \epsilon.$$
\fi
This gives an upper bound of $6cr^2 \log(1/\epsilon)$ on $N_1$.

For the induction, we use Chernoff bounds, noting that $N_{d+1}$ can be bounded as follows.  Conditioned on the event that $N_d \leq \log(1/\epsilon)(6cr^2)^d$,
we note the number of edges adjacent to nodes of distance $d$ is bounded above by the sum of $N_d$ independent binomial
random variables as above, and each such edge generates at most $r-1$ nodes for $N_{d+1}$.  Let $N'_d$ be the number of such edges.  Then we have
\begin{eqnarray*}
\Pr\left(N_{d+1} > (6cr^2)^{d+1}\log(1/\epsilon)\right) \leq\\
 \Pr\left(N_{d+1} > (6cr^2)^{d+1}\log(1/\epsilon)~|~N_{d} > (6cr^2)^{d}\log(1/\epsilon)\right)
\mbox{} +\\
 \Pr\left(N_{d+1} > (6cr^2)^{d+1}\log(1/\epsilon)~|~N_{d} \leq
 (6cr^2)^{d}\log(1/\epsilon)\right)  \leq \\
                                              d\epsilon + \Pr\left(N'_{d} > \left( (6cr^2)^{d}\cdot (6cr)\right) \log(1/\epsilon)~|~N_{d} \leq \log(1/\epsilon)(6cr^2)^{d}\right).
\end{eqnarray*}
We bound the last term via a Chernoff bound, noting that the sum of the $N_d$ independent binomial random variables {$B({n-1 \choose r-1}, q)$ has the same distribution as the sum of $N_d {n-1 \choose r-1}$ independent Bernoulli random variables that take value 1 with probability $q$}.
We use the Chernoff bound from \cite[Theorem 4.4, part 3]{MU}, which says that if $X$ is the sum of independent 0-1 trials and $E[X] = \mu$, then for $R \geq 6\mu$,
$$\Pr(X \geq R) \leq 2^{-R}.$$
Hence, 
\begin{eqnarray*}\Pr\left(N'_{d} > \log\left(1/\epsilon\right)\left(6cr^2\right)^{d} \cdot \left(6cr\right)~|~N_{d} \leq
 \log\left(1/\epsilon\right)\left(6cr^2\right)^{d}\right)\\
  \leq 2^{-\log\left(1/\epsilon\right)\left(6cr^2\right)^{d}\cdot\left(6cr\right)} \leq \epsilon,
\end{eqnarray*}
completing the induction and giving the lemma.  
\end{proof}

Let $E$ be the event that the parallel peeling process
on $G_c^r$ terminates after
$\frac{1}{\log((k-1)(r-1))}\log \log n + O(1)$ rounds.
Our goal is to show that $\Pr[E]=1-o(1)$.
Let $c_1$ any $c_2$ be the constants appearing in Lemma \ref{lem:polylog}. 
Let $E_1$ denote the event that, for all vertices $v$ in $G_c^r$, the neighborhood of distance $c_1 \log \log n$ around $v$ contains at most $\log^{c_2}n$ vertices, and let
$\bar{E}_1$ denote the event that $E_1$ does not occur.  
\begin{lemma}
\label{lemma:stupid} It holds that 
$\Pr[E] \geq \Pr[E| E_1] - 1/n$.
\end{lemma}
\begin{proof}
Note that 
\begin{equation}
\label{eq:stupid}  
\Pr[E] = \Pr[E | E_1] \Pr[E_1] + \Pr[E | \bar{E}_1] \Pr[\bar{E}_1].
\end{equation}
By Lemma \ref{lem:polylog}, $\Pr[E] \geq 1-1/n$. 
Hence, by Equation \eqref{eq:stupid},
$\Pr[E] \geq \Pr[E | E_1] (1 - 1/n) \geq \Pr[E| E_1] - 1/n$. 
\end{proof}

Lemma \ref{lemma:stupid} implies that, if we show that $\Pr[E | E_1] = 1-o(1)$, 
then $\Pr[E] = 1-o(1)$ as well. This is the task to which we now turn.

\subsubsection{Completing the Proof of Theorem \ref{thm:loglog}}
\label{sec:thm:loglog}
It will help us to introduce some terminology.  We will recursively
refer to a vertex other than the root as {\em peeled} in round $i$ if
it has fewer than $k-1$ unpeeled children edges (that is, edges to children) at the beginning of
the round; similarly, we say that an edge $e$ is peeled at round $i$ if some vertex incident to $e$ is
peeled.  We refer to an edge or vertex that is not
peeled as {\em unpeeled}.  At round $0$, all edges and vertices begin as unpeeled.  For
the root, we require there to be fewer than $k$ unpeeled children edges before it
is peeled.

\begin{proof}[Proof of Theorem \ref{thm:loglog}]
We analyze how the actual
branching process deviates from the idealized branching
process analyzed in Section \ref{sec:high}, showing the deviation leads to only lower order effects.  We view the
branching process as generating a breadth first search (BFS) tree of depth at most $O(\log \log n)$ rooted at the
initial vertex $v$.  To clarify, breadth first search trees are defined
such that once a vertex $u$ is expanded in the breadth
first search, $u$ cannot be the child of any vertex $u'$ in the tree that is expanded after $u$.


\begin{lemma}
\label{lemma:expandu}
When expanding a node $u$ in the BFS tree
rooted at vertex $v$ in $G_c^r$, let 
$Z_u$ denote the number of already expanded vertices
in the BFS tree, and let $N(u)$ denote the number of child edges of $u$ in the BFS tree.
If $Z_u = \text{polylog}(n)$, 
then
$N(u)$ is a random variable with total variation distance at most $\text{polylog}(n)/n$
from $\text{Poisson}(rc)$. 
\end{lemma}

\begin{proof}
The number of
children edges incident to $u$ in $G_c^r$ is a binomial random variable 
$B(M/q, q)$, 
where the mean $M$ equals ${{n - Z_u -1} \choose {r-1}}q$.  
Since
$Z_u$ is polylogarithmic in $n$,  

\begin{eqnarray*} M={{n - Z_u -1} \choose {r-1}}q =   {n-1 \choose r-1} q (1-\text{polylog}(n)/n) = rc (1-\text{polylog}(n)/n).\end{eqnarray*} 
We invoke Le Cam's Theorem \cite{lecam} (see Appendix \ref{app:lecam} for the statement), which bounds the total variation distance between 
binomial and Poisson distributions, to conclude that
the total variation distance between $B(M/q, q)$
and $\text{Poisson}(M)$ is at most $Mq \leq rc (cn/{n \choose r})= O(1/n^{r-1})$. Meanwhile, the total variation distance between 
$\text{Poisson}(M)$ and $\text{Poisson}(rc)$ is $\text{polylog}(n)/n$,
and so by the triangle inequality, 
the total variation distance between $\text{Poission}(rc)$ and $B(M/q, q)$
is also $\text{polylog}(n)/n$. 
\end{proof}


\begin{lemma}
\label{lemma:tvbound1}
Let $X_1(v)$ denote the random variable
describing the tree of depth $i=O(\log \log n)$ rooted at $v$ 
in the idealized branching process. Let $X_2(v)$ denote the random
variable describing
the BFS tree of depth $i$ rooted at $v$ in $G_c^r$, conditioned on event
$E_1$ occurring. 
The total variation distance between $X_1(v)$ and $X_2(v)$
is at most $\text{polylog}(n)/n$.
\end{lemma}
\begin{proof}
We describe a standard coupling of the actual branching process
and the idealized branching process. 
That is, we imagine running two different experiments $(Y_1(v), Y_2(v))$, with $Y_1(v)$ corresponding to the idealized branching process, and $Y_2(v)$ corresponding to the actual branching process conditioned on event $E_1$ occurring. 
The two branching processes will not be independent, yet $Y_1(v)$ and $Y_2(v)$ will have the same distribution as the idealized and
actual branching processes $X_1(v)$ and $X_2(v)$ respectively. We will 
show that for any $i = O(\log\log n)$, with probability at least $1-\text{polylog}(n)/n$ the two
experiments \emph{never} deviate from each other. 
It follows that any event that occurs in $X_1(v)$ with probability 
$p$ occurs in $X_2(v)$ with probability $p \pm \text{polylog}(n)/n$,
and hence the total variation distance between $X_1(v)$ and $X_2(v)$
is at most $\text{polylog}(n)/n$ as desired.

The experiments $Y_1(v)$ and $Y_2(v)$ proceed as follows. Both $Y_1(v)$
and $Y_2(v)$ begin by expanding a node $v$. 
Recall that 
the 
number of child edges of $v$ in the idealized branching process 
has distribution $\mu_{\text{ideal}}$, where $\mu_{\text{ideal}}$ denotes a discrete Poisson random variable with mean $rc$. 
Let $\mu_v$ denote the distribution of $N(v)$
in the real branching process conditioned on event $E_1$ occurring. 
Define $\alpha_v(x) = \min\{\mu_{\text{ideal}}(x), \mu_{v}(x)\}$. 

Let $\gamma_v$ denote the total variation distance between $\mu_{\text{ideal}}$
and $\mu_v$; by Lemma \ref{lemma:expandu}, $\gamma_v \leq \text{polylog}(n)/n$. 
Note that
$\sum_x \alpha_v(x) = 1-\gamma_v$, and hence $\alpha'_v=\alpha_v/(1-\gamma_v)$ is a probability
distribution. 

At the start of experiments $X_1(v)$ and $X_2(v)$, we toss a coin with a probability of heads equal to $1-\gamma_v$. 
If it comes up heads, we choose $N$ from the probability
distribution $\alpha_v'$, and set the number of child edges of $v$
in both $Y_1(v)$ and $Y_2(v)$ to be $N$, and choose identical identifiers for their children
uniformly at random from $[n] \setminus \{v\}$ without replacement. 
If it comes up tails, we choose the number of child edges of $v$ in
$Y_1(v)$
according to the probability distribution $\sigma_{\text{ideal}, v}(x)$ defined via:
\[\begin{cases} 
\frac{\mu_{\text{ideal}}(x) - \mu_v(x)}{\gamma_v} & \text{if } \mu_{\text{ideal}}(x) > \mu_v(x)\\
0 & \text{otherwise,}
\end{cases}
\]
\noindent choose the number of child edges of $v$ in $Y_2(v)$ according
to the distribution $\sigma_{\text{real}, v}(x)$ defined via:
\[\begin{cases}
 \frac{\mu_v(x)-\mu_{\text{ideal}}(x) }{\gamma_v} & \text{if }  \mu_v(x) > \mu_{\text{ideal}}(x)\\
0 & \text{otherwise,}
\end{cases}
\]
\noindent and independently choose identifiers for their children at random
from $[n]\setminus \{v\}$, without replacement.

Under these definitions, the number of child edges of $v$ in $Y_1(v)$
is distributed according to $\mu_{\text{ideal}}$, while the number of child
edges of $v$ in $Y_2(v)$ is distributed according to $\mu_v$. That is,
these quantities have the correct marginals, even though $Y_1(v)$ and $Y_2(v)$
are not independent. 

If the coin came up tails, we then run $Y_1(v)$ and $Y_2(v)$ independently of each other
for the remainder of the experiment. If the coin came up heads, 
we repeatedly expand nodes in both $X_1(v)$ and $X_2(v)$ as follows.
When expanding a node $u$, we let 
 $\mu_u$ denote the distribution of $N(u)$
in the real branching process,
and we define $\alpha_u$, $\gamma_u$, $\alpha'_u$, $\sigma_{\text{ideal}, u},$ and $\sigma_{\text{real}, u}$ analogously. 
We toss a new coin with a probability of heads equal to $1-\gamma_u$. If
the new coin comes up heads, we choose $N$  from the probability distribution
$\alpha'_u$
and set the number of child edges of $u$
in both $Y_1(v)$ and $Y_2(v)$ to be $N$, and choose identical identifiers for their children
uniformly at random from $[n] \setminus T$, where $T$ is the set of nodes already
appearing in the (identical) trees. 
If the new coin comes up tails, we choose the number of child edges of $u$ in $Y_1(v)$ according to $ \sigma_{\text{ideal}, u}$,
choose the number of child edges of $u$ in $Y_2(v)$ according to  $\sigma_{\text{real}, u}$,
and independently choose the identifiers of the children at random from the set of nodes not already appearing in the respective tree, without replacement.

It is straightforward to check that the marginal distributions of $Y_1(v)$ and $Y_2(v)$
are the same as $X_1(v)$ and $X_2(v)$. Moreover, each time a node 
$u$ is expanded in $Y_2(v)$, the processes deviate from each other
with probability at most $\gamma_u$. Since $X_2(v)$
describes the actual branching process conditioned on event $E_1$ occurring,
Lemma \ref{lemma:expandu} guarantees that 
 $\gamma_u \leq \text{polylog}(n)/n$ for all nodes $u$
 that are ever expanded. 
Moreover, at most $\text{polylog}(n)$ nodes $u$ 
are ever expanded in $Y_2(v)$. 
By the union bound over all $\text{polylog}(n)$ nodes $u$
ever expanded in $Y_2(v)$, it holds that $Y_1(v)$ and $Y_2(v)$ \emph{never}
deviate with probability at least $1-\text{polylog}(n)/n$.
\end{proof}

Recall that 
$\lambda_i$ is the probability that the root node $v$ survives after $i$ rounds of the idealized branching process. Let $\lambda_i^{(a)}$ denote the corresponding value in the actual branching process conditioned on event $E_1$ occurring. That is, 
\begin{equation} \label{eq:lambdaia} \lambda_i^{(a)}=\Pr[v \text{ survives } i \text{ rounds of peeling in } G_c^r|E_1].\end{equation}
By symmetry, the probability on the right hand side of Equation \eqref{eq:lambdaia} is independent of the node $v$.

Lemma \ref{lemma:tvbound1} implies that $\lambda_i$ and $\lambda_i^{(a)}$ differ by at most 
$\text{polylog}(n)/n$ for all $i = O(\log\log n)$, and thus 
$$\lambda_{t^*}^{(a)} \leq 
\lambda_{t^*} + \text{polylog}(n)/n \leq  \text{polylog}(n)/n.$$ 
It remains to improve the upper bound on $\lambda^{(a)}_i$ to
$o(1/n)$, as this will allow us to apply a union bound over
all the vertices $v$ to conclude that with probability $1-o(1)$, no vertex survives after $i$ rounds of peeling.  For expository purposes, we first show how to do this assuming the
neighborhood is a tree. We then show how to handle the general case,
in which vertices may be duplicated as we
expand the neighborhood of the root node $v$.  When duplicates appear, parts of our
neighborhood tree expansion are no longer independent, as in our
idealized analysis, but we are able to modify the analysis to cope with these dependencies.

\medskip
\edit{
\noindent {\bf Bounding $\lambda_i$ for Trees:}  Assume for now that the neighborhood of the root node $v$ is a
tree. Note that for the root to be unpeeled after $i$ rounds, there must be at least $k \geq 2$
adjacent unpeeled edges, corresponding to at least 2 (distinct, from
our tree assumption) unpeeled children vertices 
after $i-1$ rounds.  
We have shown that, conditioned on event $E_1$ occurring, each vertex remains unpeeled for at most $t^* =
\frac{1}{\log((k-1)(r-1))}\log \log n + O(1)$ rounds with probability
$O(\text{polylog}(n)/n)$.  The 2 unpeeled children vertices can be chosen from the at most 
polylogarithmic number of children of $v$ (the polylogarithmic bound follows
from the occurrence of event $E_1$).
This gives only ${\text{polylog}(n) \choose 2} = \text{polylog}(n)$
possible sets of choices.  Hence, via a union bound, the probability that $v$ survives at least
$t^*+1$ rounds is bounded above by $\text{polylog}(n) \cdot \left(\text{polylog}(n)/n\right)^2 = O(\text{polylog}(n)/n^2) = o(1/n)$.  We can take a union bound over all vertices
for our final $1-o(1)$ bound. 
}

\medskip
\edit{
\noindent {\bf Dealing with duplicate vertices:}  Finally, we now explain that, with probability $1-o(1)$, we need to worry
only about a single duplicate vertex in the neighborhood for all
vertices, and further that this only adds an additive constant to the
number of rounds required. Conditioned on event $E_1$
occurring, for any fixed node $v$ it holds that as we expand the neighborhood of $v$ 
of distance $O(\log \log n)$ using
breadth first search, the probability of a duplicate vertex occurring during any expansion step is only 
$\text{polylog}(n)/n$.  As the neighborhood contains only a polylogarithmic number of vertices,
the probability of having at least two duplicate vertices within the neighborhood of $v$ is
$o(1/n)$. By a union bound over all $n$ nodes $v$, with probability $1-o(1)$, 
\emph{no} node $v$ in the graph will have two duplicated vertices in the BFS tree
rooted at $v$. We 
refer to this event as $E_2$, and we condition
on this event occurring for the remainder of the proof. This conditioning
does not affect our estimate of $\Pr[E | E_1]$ by more than an additive $o(1)$ factor, for the same reason
conditioning on $E_1$ did not affect our estimate of $\Pr[E]$ by more than an additive $o(1)$
factor (cf. Lemma \ref{lemma:stupid}). Indeed, 
\begin{eqnarray*}
\Pr[E|E_1] = \Pr[E | E_1 \cap E_2] \Pr[E_2] + \Pr[E | E_1 \cap \bar{E}_2] \Pr[\bar{E}_2]\\
\geq \Pr[E|E_1 \cap E_2] (1-o(1)).
\end{eqnarray*}
It is therefore sufficient to show that, conditioned on event $E_1$ occurring,
having one duplicate vertex in the neighborhood only adds a constant number of rounds to the parallel
peeling process. 

We first consider the case when $r \geq 3$, so that if the root
remains unpeeled it has at least four (not necessarily distinct)
unpeeled vertices at distance $1$ from it, corresponding to the
at least two edges (each with at least two other vertices, as $r \geq 3$) that prevent the
root from being peeled.  If we encounter a duplicate vertex, we
pessimistically assume that it prevents two vertices adjacent to the
root -- namely, its ancestors -- from being peeled.  
Even with this pessimistic assumption, 
simply adding one additional layer of expansion in the neighborhood
allows the root to be peeled by round $t^* + 2$ with probability
$1-o(1/n)$, as we now show.

Consider what happens in $t^* + 2$ rounds when there is 1 duplicate
vertex. As stated in the previous paragraph, for the root to remain unpeeled, 
it must have at least four neighbors, and at most two of
these four vertices is a duplicate or has a descendant that is a
duplicate. Thus, in order for the root to remain unpeeled after $t^* + 2$ rounds, 
at least two neighbors, $u_1$ and $u_2$, of the root must remain unpeeled after $t^* + 1$ rounds, when the neighborhoods of
$u_1$ and $u_2$ for $t^*+1$ rounds are trees. By our previous calculations,
the probability that $u_1$ and $u_2$ both remain unpeeled after $t^*+1$ rounds
when their neighborhoods are trees is $O(\text{polylog}(n)/n^2)$. Thus,
we take a union bound over the at most $\text{polylog}(n)$ pairs of descendants of the root, and conclude that the probability that the root survives $t^*+2$ rounds of the peeling process is
$1-o(1/n)$. 

Finally, union bounding over all nodes $v$ in $G_c^r$, 
we conclude that \emph{all} nodes in $G_c^r$ are peeled after $t^*+2$ rounds
 with probability $1-o(1)$. That is, we have shown that $\Pr[E | E_1] = 1-o(1)$.

The case where $r=2$ and $k \geq 3$ requires a bit more care.   Let us consider what happens after
$t^*+3$ rounds in this case.  For 
the root note $v$ to remain unpeeled, $v$ must have at least $k \geq 3$ 
incident edges that remain unpeeled after $t^* + 2$ rounds of peeling.
This 
corresponds to at least $3$ (not necessarily distinct) unpeeled children of $v$.
Thus, even if there is one duplicate vertex in the neighborhood of $v$, $v$ must have at least
one unpeeled child $u$ whose neighborhood of distance $t^*+2$ is a tree. 
This vertex must have at least two children (grandchildren of the root) 
that must remain unpeeled for $t^*+1$ rounds. 
Thus, by our previous calculations, the probability that $u$
 remains unpeeled after $t^*+2$ rounds is at most
$\text{polylog}(n)/n^2$. Again we can union bound over
the at most $\text{polylog}(n)$ children $u$ of the root node $v$ to obtain a $1-o(1/n)$ probability
that $v$ remains unpeeled after $t^*+3$ rounds in this case.  

We have shown that $\Pr[E | E_1] = 1-o(1)$, and by Equation \eqref{eq:stupid},
it follows that $\Pr[E] = 1-o(1)$ as well.
}
\end{proof}

\noindent {\em Remark:} One can obtain better than $1-o(1)$ bounds on the probability of 
terminating after $\frac{1}{\log((k-1)(r-1))}\log \log n + O(1)$ rounds when $c < c^*_{k, r}$.
For example, $1-o(1/n)$ bounds are possible when $r > 3$;  the argument requires 
considering cases for the possibility
that 2 vertices are duplicated in the neighborhood
around a vertex.  However, one cannot hope for probability bounds 
of $1-o(1/n^a)$ for an arbitrary constant $a$ when duplicate edges
may appear, as is typical for hashing applications.
The probability the $k$-core
is not empty because $k$ edges share the same $r$ vertices is 
$\Omega(n^{-kr+k+r})$ for constant $k$, $r$, and graphs with a linear
number of edges, which is already $\Omega(1/n)$ for $k=2$ and $r=3$ 
\edit{or for $k=3$ and $r=2$.}  

\editmm{
\subsubsection{Completing the Proof of Theorem \ref{thm:logloglb}}
\label{sec:thm:loglogb}
Recall that Theorem \ref{thm:logloglb} claims
that with probability $1-o(1)$, at least $\frac{1}{\log ((k-1)(r-1))}\log \log n - O(1)$ rounds of peeling
are required before arriving at an empty $k$-core. The analysis of Section \ref{sec:high} established
that, in the idealized setting, each node $v$ remains unpeeled after $t=\frac{1}{\log ((k-1)(r-1))}\log \log n - C_1$
rounds with probability at least $n^{-1/3}$, where where $C_1$ is an appropriately large constant
that depends on $k$ and $r$. 
Hence, in the idealized setting, the expected number of nodes that remain unpeeled after $t$ rounds 
is greater than or equal to $n^{2/3}$.
We use this fact to establish that 
the claimed round lower bound holds in $G_c^r$ with probability $1-o(1)$.

The argument to bound the effects of deviations 
from the idealized process 
is substantially simpler in the context of Theorem \ref{thm:logloglb} 
than in the analogous argument from Section \ref{sec:thm:loglog}.
Indeed, 
to prove Theorem \ref{thm:loglog}, we needed to establish that with probability $1-o(1)$, 
\emph{all} nodes in $G_c^r$ are peeled 
after a suitable number of rounds. The argument of Section \ref{sec:thm:loglog}
accomplished this by establishing that, for any node $v$,
$v$ is peeled after $t$ rounds with probability $1-o(1/n)$, for
an appropriate choice of  $t=\frac{1}{\log ((k-1)(r-1))}\log \log n + O(1)$.
We then
applied a union bound to conclude that this holds for
all nodes with probability $1-o(1)$. It was relatively easy
to establish that $v$ is peeled after $t$ rounds
with probability $1-\text{polylog}(n)/n$, and most of the effort in the proof 
was devoted to increasing this probability to $1-o(1/n)$, large enough to perform a union bound
over all $n$ nodes. 

In contrast, to establish a \emph{lower} bound on the number of rounds required, 
one merely needs to show the existence of a single node
that remains unpeeled after $t=\frac{1}{\log ((k-1)(r-1))}\log \log n - C_1$ rounds. Let $L_{t, \text{ideal}}$
be a random variable denoting the number of nodes that remain unpeeled after $t$ rounds 
in the idealized setting of Section \ref{sec:high},
and let $L_t$ be a random variable denoting the analogous number of nodes 
in $G_c^r$. As previously mentioned,
our analysis in the idealized framework (Section \ref{sec:high})
shows that 
the expected value of $L_{t, \text{ideal}}$ is at least $n^{2/3}$ for a suitably
chosen constant $C_1$ in the expression for $t$.  
Lemma \ref{lemma:tvbound1} then implies that the expected value of $L_t$
is at least $n^{2/3}/\text{polylog}(n)$. 
We now sketch an argument that $L_t$ is concentrated around its expectation, i.e., that with probability $1-o(1)$,
$L_t= E[L_t] \pm n^{1/2} \text{polylog}(n) \geq n^{2/3}/\text{polylog}(n)$.
We note that an entirely analogous argument is used later to prove Theorem \ref{thm:abovethreshold}
in Section \ref{sec:above}, where the argument is given in full detail.

Let $E_1$ denote the event that there are $m=cn \pm O(\sqrt{n \log n})$ edges in $G_c^r$. Let $E_2$
denote the event that all nodes in $G_c^r$ have neighbors of size at most $\log^{c_2}(n)$ for an
appropriate constant $c_2$.
By Lemma \ref{lem:polylog}, events $E_1$ and $E_2$ both occur with probability $1-2/n$.
We will condition on both events occurring for the duration of the argument, absorbing
an additive $2/n$ into the $o(1)$ failure probability in the statement of Theorem \ref{thm:logloglb}
(note that the conditioning causes at most an $O(1)$ change in $E[L_t]$).

We consider the process of exposing the $m$ edges of $G_c^r$ one at a time;  \edit{denote the random edges by $A_1,A_2,\ldots,A_{m}$.}
For our martingale, we consider random variables $L_t^{i} = E[L_t~|~A_1,\ldots,A_i]$, so $L_t^0 = E[L_t]$ and $L_t^{m} = L_t$.  
Conditioned on events $E_1$ and $E_2$ occurring, 
each exposed edge changes the conditional expectation
of $L_t$ by only $\log^{c_2}(n)$,  so Azuma's martingale inequality\footnote{Formally,
to cope with conditioning on events $E_1$ and $E_2$ in the application of Azuma's inequality,
we must actually consider a slightly modified martingale. This technique is standard, and the details
can be found in Section \ref{sec:above}.}
 \cite[Theorem 12.4]{MU} yields for sufficiently large $n$:
\begin{eqnarray*}\Pr(|L_t - E[L_t]| \geq n^{1/2} \log^{c_2 + 1}(n)) & \leq & 2e^{-n \log^{2c_2+2}(n)/\left(2 m \log^{2c_2}(n)\right)}\\
 & \leq & e^{-\log^{3/2}(n)} \leq 1/n.\end{eqnarray*}
 In particular, this means that with probability
$1-o(1)$ there remain unpeeled vertices in $G_c^r$ after $t$ rounds of peeling.}

\section{Above the Threshold}
\label{sec:above}
We now consider the case when $c > c_{k,r}^*$.  We show that parallel peeling
requires $\Omega(\log n)$ rounds in this case.

Molloy \cite{Molloy} showed that in this case
there exists a $\rho > 0$ such that $\lim_{t \to \infty}\rho_t =
\rho$. Similarly, $\lim_{t \to \infty}\beta_t = \beta > 0$ and
$\lim_{t \to \infty}\lambda_t = \lambda > 0$.  It follows that the
core will have size $\lambda n + o(n)$. We examine how 
$\beta_t$ and $\lambda_t$ approach their limiting values to show that the
parallel peeling algorithm takes $\Omega(\log n)$ rounds.

\begin{theorem}
\label{thm:abovethreshold}
Let $r \geq 3$ and $k \geq 2$.
With probability $1-o(1)$, the peeling process for the $k$-core in 
$G_{n, cn}^r$
terminates after $\Omega(\log n)$ rounds when $c > c^*_{k, r}$,
\end{theorem}

\begin{proof}
First, note that $\beta$ corresponds to the fixed point 
\begin{align}
\label{eq:fixedpoint}
\beta &= \bigg[1 - e^{-\beta}\sum_{j=0}^{k-2}\frac{\beta^j}{j!}\bigg]^{r-1}rc.
\end{align}
Let $\beta_i = \beta + \delta_i$, where $\delta_i > 0$. 
We begin by working in the idealized branching process model given in Section~\ref{sec:high} to determine
the behavior of $\beta_i$. Starting with Equation \eqref{eq:betairecursion} and considering $\beta_{i+1}$ as a
function of $\delta_i$, we obtain:
\begin{align} \label{eq:betaabove}
\beta_{i+1} &= \bigg[1 - e^{-\beta-\delta_i}\sum_{j=0}^{k-2}\frac{(\beta+\delta_i)^j}{j!}\bigg]^{r-1}rc.
\end{align}

We now view the right hand side of Equation \eqref{eq:betaabove}
as a function of $\delta_i$. Denoting this function as $f(\delta_i)$,
we take a Taylor series expansion around 0 and conclude that:
$$f(\delta_i) = f(0) + f'(0) \delta_i + \Theta(f''(0) \delta_i^2).$$

Equation \eqref{eq:fixedpoint} immediately implies that $f(0)=\beta$. 
Moreover, it can be calculated that  
 \begin{equation}
\label{eq:a1} f'(0) = \frac{(r-1)\beta e^{-\beta}}{1-e^{-\beta}\sum_{j=0}^{k-2}\frac{\beta^j}{j!}}\frac{\beta^{k-2}}{(k-2)!}\end{equation} 

In particular, it holds that 
 \begin{equation}
\label{eq:a} 0 < f'(0) < 1. \end{equation} 

Note that while $f'(0) < 1$ can be checked explicitly, this 
condition also follows immediately from the convergence of the $\beta_i$
values to $\beta$.  

The fact that $0 < f'(0)$ is critical in our analysis. Indeed, when $c$ is below
the threshold density $c_{k, r}^*$, $\beta=0$, and hence Equation \eqref{eq:a1} implies that $f'(0)=0$. This is precisely why our analysis
here ``breaks'' when $c < c_{k, r}^*$, and offers an intuitive explanation
for why the number of rounds is $O(\log\log n)$ when $c < c_{k, r}^*$, but is $\Omega(\log n)$ when $c > c_{k, r}^*$.

Since $\beta_{i+1} = \beta + \delta_{i+1}$, $\delta_i$ decreases by a factor of at most $f'(0) + O(\delta_i)$ each iteration.
In particular, for small enough $\delta_i$, $\delta_i$ decreases by a factor of at most $f'(0)+\epsilon_1$ 
for some $\epsilon_1 > 0$ each iteration.

Next, we know that $\lambda = 1-e^{-\beta}\sum_{j=0}^{k-1}\frac{\beta^j}{j!}$. Equations \eqref{eq:lambdairecursion} and \eqref{eq:a}, imply that
\begin{align*}
\lambda_i &= \lambda + \frac{e^{-\beta}\beta^{k-1}}{(k-1)!}\delta_i + O(\delta_i^2).
\end{align*}
Hence, for suitably small (constant) $\delta_i$ values, in each round $\lambda_i$
gets closer to $\lambda$ by at most a constant factor under the idealized model. 
This suggests the $\Omega(\log n)$ bound.  
Specifically, we can choose $t = \gamma \log n$ for a suitably small constant $\gamma$
so that $\delta_t$ in the idealized model remains $\Omega(n^{1-\eta})$ for \edit{a given constant} $\eta < 1$.  This gives that the ``gap''
$\lambda_t - \lambda$ is $\Omega(n^{-\eta})$, leaving an expected $\Omega(n^{1-\eta})$ vertices still to be peeled.
This number is high enough so that we can apply martingale concentration arguments, as deviations from
the expectation can be made to be $o(n^{1-\eta})$ with high probability.  This follows the approach of e.g. \cite{BFU,Molloy}.

To this end, note that it is straightforward to modify the argument of Lemma~\ref{lem:polylog} to show
that for a suitably small constant $c_1 > 0$, with probability $1-O(1/n)$,
for all vertices $v$, the neighborhood of distance $c_1 \log n$ around
$v$ contains at most $n^{c_2}$ vertices for a suitable
constant $c_2>0$.  For suitable constants $c_1, c_2$, we refer to this event as $E_3$,
and we condition on $E_3$ occurring for the duration of the proof.

As before, there are deviations from the
idealized branching process, and we bound the effects
of these deviations as follows. If we let $Z_u$ be the number of already
expanded vertices in the breadth first search when expanding a vertex
$u$'s neighborhood up to distance $c_1 \log n$, we have $Z_u \leq
n^{c_2}$, so as we expand a neighborhood the probability of any
collision is at most $n^{2c_2-1}$.  Since we are proving a lower bound
on the number of rounds required, we can pessimistically assume that such vertices
(i.e., vertices $u$ such that the BFS rooted at $u$ results in a collision) will be peeled immediately -- 
this will not affect our conclusion that $\Omega(n^{1-\eta})$ vertices remain to be peeled, as we may choose $c_2$ so that 
$n^{2c_2} = o(n^{1-\eta})$.  Now we apply Azuma's martingale inequality
\cite[Theorem 12.4]{MU},
exposing the $cn$ edges in the graph one at a time;  \edit{denote the random edges by $A_1,A_2,\ldots,A_{cn}$}.  We consider $t= c_1 \log n$
 rounds for a $c_1$ that leaves a gap of \edit{$\Omega(n^{\eta})$} for some small $\eta>0$
{(i.e., guarantees that $\lambda_t - \lambda > n^{-\eta}$}; $\eta = 0.01$ suffices), and let $X_t$ be the number of vertices that survive that many rounds 
with no duplicates in their neighborhood of depth $c_1 \log n$.  Then 
$E[X_t] -\lambda n$ is $\Omega(n^{1-\eta})$.  
\edit{

For our martingale, we consider random variables $X_t^{i} = E[X_t~|~A_1,\ldots,A_i]$, so $X_t^0 = E[X_t]$ and $X_t^{cn} = X_t$.  
To cope with the conditioining on $E_3$, we consider the ancillary random variable $Y_t^i$ where $Y_t^i = X_t$ as long 
there is no neighborhood of distance $c_1 \log n$ around any vertex $v$ that
contains at most $n^{c_2}$ vertices among the $i$ currently revealed edges and $Y_t^i = Y_t^{i-1}$ otherwise, for our
suitably chosen constant $c_2$.  Note $Y_t^0= E[X_t]+O(1)$, and $\Pr(Y_t^{cn} \neq X_t^{cn})$ corresponds to the event $E_3$.\footnote{This
method of dealing with conditioning while applying Azuma's martingale inequality is well known;  see for example \cite{Fan}.}
Each exposed edge changes the conditional expectation
of $Y_t$ by only $O(n^{c_2})$ vertices, so Azuma's martinagle inequality yields:
$$\Pr(|Y_t - Y_t^0| \geq n^{2/3}) \leq 2e^{-n^{4/3}/(cn\cdot n^{2c_2})} \leq e^{-n^{1/6}}$$
for $c_2$ chosen suitably small.  This implies
$$\Pr(|X_t - E[X_t]| \geq n^{2/3}+O(1)) \leq e^{-n^{1/6}} + \Pr(E_3).$$
\noindent Hence with probability $1-o(1)$ there remain vertices to be peeled after 
$\Omega(\log n)$ rounds.
}
\end{proof}

\noindent {\em Remark:} 
As discussed in the introduction, the lower bound of Theorem \ref{thm:abovethreshold} matches an $O(\log n)$ upper bound of Achiloptas and Molloy \cite{achlio}.



\section{Simulation Results}
\label{sec:experiment}
We implemented a simulation of the parallel peeling algorithm using the $G_{n, cn}^r$ model, in order to determine how well our 
 theoretical analysis matches the empirical evolution of 
 the peeling process.  Our results demonstrate that
 the theoretical analysis matches the empirical evolution remarkably well.
 
 \begin{table*}[t]
\small
\centering
\begin{tabular}{|l|r|r|r|r|r|r|r|r|}
\hline
& \multicolumn{2}{|c|}{$c = 0.7$} & \multicolumn{2}{|c|}{$c = 0.75$} & \multicolumn{2}{|c|}{$c = 0.8$} & \multicolumn{2}{|c|}{$c = 0.85$}\\
\hline
$n$ & Failed & Rounds & Failed & Rounds & Failed & Rounds & Failed & Rounds \\
\hline
10000 & 0 & 12.504 & 0 & 23.352 & 1000 & 17.037 & 1000 & 10.773 \\
20000 & 0 & 12.594 & 0 & 23.433 & 1000 & 19.028 & 1000 & 11.928 \\
40000 & 0 & 12.791 & 0 & 23.343 & 1000 & 20.961 & 1000 & 12.992 \\
80000 & 0 & 12.939 & 0 & 23.372 & 1000 & 22.959 & 1000 & 14.104 \\
160000 & 0 & 12.983 & 0 & 23.421 & 1000 & 25.066 & 1000 & 15.005 \\
320000 & 0 & 13.000 & 0 & 23.491 & 1000 & 27.089 & 1000 & 16.305 \\
640000 & 0 & 13.000 & 0 & 23.564 & 1000 & 29.281 & 1000 & 17.334 \\
1280000 & 0 & 13.000 & 0 & 23.716 & 1000 & 31.037 & 1000 & 18.499 \\
2560000 & 0 & 13.000 & 0 & 23.840 & 1000 & 33.172 & 1000 & 19.570 \\
\hline
\end{tabular} \\
\caption{Results from simulations of the parallel peeling process using $r = 4$ and $k=2$, averaged over $1000$ trials.}
\label{table:simulation1}
\end{table*}

To check the growth of the number of rounds as a function of $n$, we ran the program $1000$ times for $r=4, k=2$ and various values of $n$ and $c$, and computed the average number of rounds for the peeling process to complete. For reference, $c_{2,4}^* \approx 0.772$. Table \ref{table:simulation1} shows the results.



For all the experiments, when $c < c_{2,4}^*$, all $1000$ trials
succeeded (empty $k$-core) and when $c > c_{2,4}^*$, all $1000$ trials
failed (non-empty $k$-core).  
For $c < c_{2,4}^*$, the
average number of rounds increases very slowly with $n$, while for $c
> c_{2,4}^*$, the average increases approximately linearly in $\log n$. This
is in accord with our $O(\log \log n)$ result below the threshold and
$\Omega(\log n)$ result above the threshold. The results for other
values of $r$ and $k$ were similar.

\begin{table*}[t]
\small
\centering
\begin{tabular} {|l|r|r|}
\hline
\multicolumn{3}{|c|}{$c = 0.7$} \\
\hline
$t$ & Prediction & Experiment \\
\hline
1 & 768922 & 768925 \\
2 & 673647 & 673664 \\
3 & 608076 & 608097 \\
4 & 553064 & 553091 \\
5 & 500466 & 500503 \\
6 & 444828 & 444872 \\
7 & 380873 & 380930 \\
8 & 302531 & 302607 \\
9 & 204442 & 204550 \\
10 & 93245 & 93398 \\
11 & 14159 & 14269 \\
12 & 74 & 78 \\
13 & 0.00001 & 0 \\
14 & 0 & 0 \\
15 & 0 & 0 \\
16 & 0 & 0 \\
17 & 0 & 0 \\
18 & 0 & 0 \\
19 & 0 & 0 \\
20 & 0 & 0 \\
\hline
\end{tabular}
\quad
\begin{tabular} {|l|r|r|}
\hline
\multicolumn{3}{|c|}{$c = 0.85$} \\
\hline
$t$ & Prediction & Experiment \\
\hline
1 & 853158 & 853172 \\
2 & 811184 & 811200 \\
3 & 793026 & 793042 \\
4 & 784269 & 784281 \\
5 & 779841 & 779851 \\
6 & 777550 & 777559 \\
7 & 776350 & 776359 \\
8 & 775719 & 775728 \\
9 & 775385 & 775394 \\
10 & 775209 & 775218 \\
11 & 775115 & 775124 \\
12 & 775066 & 775074 \\
13 & 775039 & 775048 \\
14 & 775025 & 775034 \\
15 & 775018 & 775026 \\
16 & 775014 & 775022 \\
17 & 775012 & 775020 \\
18 & 775011 & 775019 \\
19 & 775010 & 775018 \\
20 & 775010 & 775018 \\
\hline
\end{tabular}
\caption{Simulation results evaluating how well Equation \eqref{eq:lambdai} approximates the number of vertices left after $t$ rounds. The experiments are run using $r=4, k=2, n = 1$ million, averaged over $1000$ trials.}
\label{table:simulation2}
\end{table*}


We also tested how well the idealized values from
the recurrence for $\lambda_t$ (Equation \eqref{eq:lambdai}) approximate the fraction of vertices
left after $t$ rounds. Table \ref{table:simulation2} shows that the recurrence indeed describes the
behavior of the peeling process remarkably well, both below and above the
threshold. In these simulations, we used $r = 4, k = 2$ and
$n = 1$ million.  For each value of $c$, we averaged over $1000$ trials.

\section{GPU Implementation}
\label{sec:GPU}
\noindent \textbf{Motivation.} Using a graphics processing unit (GPU), we developed a parallel implementation for Invertible Bloom Lookup Tables (IBLTs), a data structure recently proposed by Goodrich and Mitzenmacher \cite{GM}. Two motivating applications are sparse recovery \cite{GM} and efficiently encodable and decodable error correcting codes \cite{MV}.  For brevity we describe here only the sparse recovery application.

In the sparse recovery problem, $N$ items are inserted into a set $S$,
and subsequently all but $n$ of the items are deleted. The goal is to
recover the exact set $S$, using space proportional to the final
number of items $n$, which can be much smaller than the total number
of items $N$ that were ever inserted. IBLTs achieve this roughly as
follows. The IBLT maintains $O(n)$ cells, where each cell contains
a key field and a checksum field. We use $r$ hash functions
$h_1, \dots, h_r$. When an item $x$ is inserted or deleted from $S$,
we consider the $r$ cells $h_1(x) \dots h_r(x)$, and we XOR the key
field of each of these cells with $x$, and we XOR the checksum field
of each of these cells with $\text{checkSum}(x)$, where
$\text{checkSum}$ is some simple pseudorandom function. Notice that
the insertion and deletion procedures are identical.

In order to recover the set $S$, we iteratively look for ``pure'' cells -- these are cells that only contain one item $x$ in the final set $S$.  Every time we find a pure cell whose key field is $x$, we recover $x$ and delete $x$ from $S$, which hopefully creates new pure cells. We continue until there are no more pure cells, or we have fully recovered the set $S$. 

The IBLT defines a random $r$-uniform hypergraph $G$, in which vertices correspond to cells in the IBLT, and edges correspond to items in the set $S$. Pure cells in the IBLT correspond to vertices of degree less than $k=2$. The IBLT recovery procedure precisely corresponds to 
a peeling process on $G$, and the recovery procedure is
successful if and only if the 2-core of $G$ is empty.  

We note that this example application is similar to other applications of peeling algorithms.  For example, in the setting of erasure-correcting
codes \cite{LMSS}, encoded symbols correspond to an XOR of some number of original message symbols.  This naturally defines a hypergraph in which vertices correspond to encoded symbols,
edges correspond to unrecovered original message symbols, and a vertex can recover a message symbol when its degree is 1.  
Decoding of this erasure-correcting code corresponds to peeling 
on the associated hypergraph (after deleting all vertices corresponding to erased codeword symbols), and full recovery of the message occurs when the 2-core is empty. Our analysis directly applies
to the setting where each message symbol randomly chooses to 
contribute to a fixed number $r$ of encoded symbols.    

\medskip
\noindent \textbf{Implementation Details.}
\label{sec:impelementationdetails}
Our parallel IBLT implementation consists of two stages: the insertion/deletion stage, during which items are inserted and deleted from the IBLT, and the recovery phase.
Both phases can be parallelized. 

One method of parallelizing the insertion/deletion phase is as follows: we devote a separate thread to each item to be inserted or deleted. A caveat is that multiple threads may try to modify a single cell at any point in time,
and so we have to use atomic XOR operations, to ensure that threads trying to write to the same cell do not interfere with each other. In general, atomic operations can be a bottleneck in any parallel implementation; if $t$ threads try to write to the same memory location, the algorithm will take at least $t$ (serial) time steps. Nonetheless, our experiments showed this parallelization technique to be effective.

We parallelize the recovery phase as follows. We proceed in rounds, and in each round we devote a single thread to each cell in the IBLT. 
Each thread checks if its cell is pure, and if so it identifies the item contained in the cell, removes all $r$ occurrences of the item from the IBLT, and marks the cell as recovered. 
The implementation proceeds until it reaches an iteration where no items are recovered -- this can be checked by summing up (in parallel) the number of cells marked recovered after each round, and stopping
when this number does not change.
This procedure also requires atomic XOR operations, as two threads may simultaneously try to write to the same cell if there are two or more items $x \neq y$ recovered in the same round
such that $h_i(x) = h_i(y)$ for some $1 \leq i \leq r$. 

In addition, we must take care to avoid deleting an item multiple times from the IBLT.
Indeed, since any item $x$ inserted into the IBLT is placed into $r$ cells, $x$ might be contained in multiple pure cells at any instant, and the thread devoted to each such pure cell may try to delete $x$.  This issue is not specific to the IBLT application: 
 \emph{any}
 implementation of the parallel peeling algorithm on a hypergraph,
 regardless of the application domain, must avoid peeling the same edge from the hypergraph
 multiple times. 
 
To prevent this, we split the IBLT up into $r$ subtables, and hash each item into one cell in each subtable upon insertion and deletion. When we execute the recovery algorithm, we iterate through the subtables serially (which requires $r$ serial steps per round), processing each subtable in parallel. This ensures that an item $x$ only gets removed from the table once, since the first time a pure cell is found containing $x$, $x$ gets removed from all the other subtables.

This recovery procedure corresponds to an interesting and fundamental variant of the peeling process we analyze formally in 
Appendix \ref{sec:peelsubtable}.  
In particular, one might initially expect that the number of (parallel)
time steps required by our recovery procedure may be $r$
times larger than the peeling process analyzed in Section \ref{sec:belowthresh}, since our IBLT implementation requires $r$ serial steps to iterate through all $r$ subtables. However, we prove
that the total number of parallel steps
required by our IBLT implementation is roughly a factor of $\log_2(r-1)$
larger than the $\frac{1}{\log((k-1)(r-1))} \log\log n + O(1)$ bound
proved for the peeling process of Section \ref{sec:belowthresh}. 
This ensures that, in practice, the need to iterate serially through subtables does not create a significant serial bottleneck.
Our analysis is connected in spirit to V\"{o}cking's work on asymmetric
load balancing \cite{vocking}, and we provide detailed discussion on the comparison between Theorems \ref{thm:loglog} and \ref{thm:informal} in 
Appendix \ref{sec:peelsubtable}.

\begin{theorem}(Informal) \label{thm:informal}
Let $r\geq 3$, and $\phi_{r-1} = \lim_{k \rightarrow \infty} F_{r-1}^{1/k}(k)$
be the growth rate for the Fibonacci sequence of order $r-1$.
For $c < c_{k, r}^*$,
peeling with sub-tables on $G^r_{n, cn}$ terminates after $\frac{r}{r \log \phi_{r-1} + \log(k-1)} + O(1)$ sub-rounds. \end{theorem}

\edit{We remark that while Theorem \ref{thm:loglog} holds for $r=2$, $k \geq 3$,  Theorem \ref{thm:informal}
holds only for  $r\geq 3$.}

\medskip
\noindent \textbf{Experimental Results.}
All of our serial code was written in C++ and all experiments were compiled with g++ using the -O3 compiler optimization flag and run on a workstation with a 64-bit Intel Xeon architecture and 48 GBs of RAM. We implemented all of our GPU code in CUDA with all compiler optimizations turned on, and ran our GPU implementation on an NVIDIA Tesla C2070 GPU with 6 GBs of device memory.

\medskip \noindent \textbf{Summary of results.} Relative to our serial implementation, our GPU implementation achieves 10x-12x speedups for the\\ insertion/deletion phase, and 20x speedups for the 
recovery stage when the edge density of the hypergraph is below the threshold for successful recovery (i.e. empty 2-core). When the edge density is slightly above the threshold for successful recovery, our parallel
recovery implementation was only about 7x faster than our serial implementation. The reasons for this are two-fold. Firstly, above the threshold, many more rounds of the parallel peeling process were necessary before the 2-core was found. Secondly, above the threshold, less work was required of the serial implementation because fewer items were recovered;  in contrast, the parallel implementation examines every cell in every round. 

Our detailed experimental results are given in Tables \ref{table:justinexpts} (for the case of $r=3$ hash functions) and \ref{table:justinexpts2} (for the case of $r=4$ hash functions). 
The timing results are averages over 10 trials each. For the GPU implementation,
the reported times do count for the time to transfer data (i.e. the items to be inserted) from the CPU to the GPU. 

The reported results are for a fixed IBLT size, consisting of $2^{24}$ cells. These results are representative for all sufficiently large input sizes: once the number of IBLT cells is larger than about $2^{19}$,
the runtime of our parallel implementation grows roughly linearly with the number of table cells (for any fixed table load). Here, table load refers to the ratio of the number of items in the IBLT to the number of cells in the IBLT. This corresponds to the edge density $c$ in the corresponding hypergraph.
 The linear increase in runtime above a certain input size is typical, and is due to the fact that there is a finite number of threads that the GPU can launch
at any one time. 

\begin{table*}[t]
\footnotesize
\centering
\begin{tabular} {|c|c|c|c|c|c|c|}
\hline
Table & No. Table &  \% & GPU & Serial & GPU & Serial \\
Load &	 Cells &  Recovered & Recovery Time & Recovery Time & Insert Time & Insert Time\\
\hline
0.75 & 16.8 million & 100\% & 0.33 s & 6.37 s & 0.31 s & 3.91 s\\
\hline
0.83 & 16.8 million & 50.1\% & 0.42 s & 3.64 s & 0.35 s & 4.34 s\\
\hline
\end{tabular}
\caption{Results of our parallel and serial IBLT implementations with $r=3$ hash functions. The table load refers to the ratio of the number of items in the IBLT to the number of cells in the IBLT. }
\label{table:justinexpts}
\end{table*}

\begin{table*}[t]
\footnotesize
\centering
\begin{tabular} {|c|c|c|c|c|c|c|}
\hline
Table & No. Table & \% & GPU & Serial & GPU & Serial \\
 Load &	 Cells &  Recovered & Recovery Time & Recovery Time & Insert Time & Insert Time\\
\hline
0.75 & 16.8 million & 100\% & 0.47 s & 8.37 s & 0.42 s& 4.55 s\\
\hline
0.83 & 16.8 million & 24.6\% & 0.25 s & 2.28 s & 0.46 s & 5.0 s\\
\hline
\end{tabular}
\caption{Results of our parallel and serial IBLT implementations with $r=4$ hash functions. The table load refers to the ratio of the number of items in the IBLT to the number of cells in the IBLT. }
\label{table:justinexpts2}
\end{table*}

\section{Rounds as a Function of the Distance from the Threshold}
\label{sec:distance}
Recall that the hidden constant in the $O(1)$ term of Theorem \ref{thm:loglog} depends on the size of the ``gap'' $\nu=c_{k,r}^*-c$ between the edge density and the threshold density. \edit{This 
term can be significant in practice when $\nu$ is small, and 
in this section, we make the dependence on $\nu$ explicit.} Specifically, we extend the analysis of Section \ref{sec:belowthresh} 
to characterize how the growth of the number of rounds depends on  $c_{k,r}^*-c$, when $c$ is a constant with $c < c_{k, r}^*$. 
The proof of Theorem \ref{thm:distance} below is in Appendix \ref{app:distance}.

\begin{theorem} \label{thm:distance} Let $\nu = |c_{k, r}^*-c|$ for constant $c$ with $c < c_{k,r}$. With probability $1-o(1)$, peeling in $G_{n, cn}^r$ requires $\Theta(\sqrt{1/\nu}) + \frac{1}{\log((k-1)(r-1))}\log \log n$ rounds when $c$ is below the threshold density $c_{k, r}^*$. 
\end{theorem}


\section{Conclusion}
\label{sec:conclusion}
In this paper, we analyzed parallel versions of the peeling process on random hypergraphs. We showed that when the number of edges is below the threshold edge density for the $k$-core to be empty, with high probability the parallel algorithm takes $O(\log \log n)$ rounds to peel the $k$-core to empty. In contrast, when the number of edges is above the threshold, with high probability it takes $\Omega(\log n)$ rounds for the algorithm to terminate with a non-empty $k$-core. We also considered some of the details of implementation and proposed a variant of the parallel algorithm that avoids a fundamental implementation issue;  specifically, by using subtables, we avoid peeling the same element multiple times. We show this variant converges 
significantly faster than might be expected, thereby avoiding a sequential bottleneck. Our experiments confirm our theoretical results and show that in practice, peeling in parallel provides a considerable increase in efficiency over the serialized version.

\appendix

\section{Le Cam's Theorem}
\label{app:lecam}
Le Cam's Theorem can be stated as follows.
\begin{theorem}
Let $X_1,X_2,\ldots,X_n$ be independent 0-1 random variables with $\Pr(X_i = 1) = p_i$.
Let $\lambda = \sum_{i=1}^n p_i$ and $S= \sum_{i=1}^n X_i$.  Then 
$$\sum_{k=0}^\infty |\Pr(S = k) - e^{-\lambda} \lambda^k / k!| < 2 \sum_{i=1}^n p_i^2.$$
\end{theorem}
In particular, when $p_i = \lambda/n$ for all $i$, we obtain that the binomial distribution
converges to the Poisson distribution, with total variation distance bounded by $\lambda^2/n$.  


\section{Parallel Peeling with Subtables}
\label{sec:peelsubtable}
The parallel peeling process used in our GPU implementation of IBLTs in Section \ref{sec:GPU} does not precisely correspond 
to the one analyzed in Sections \ref{sec:compl} and \ref{sec:above}. The differences are two-fold. First, the underlying hypergraph $G$
in our IBLT implementation is not chosen uniformly from all $r$-uniform hypergraphs; instead, vertices in $G$ (i.e., IBLT cells) are partitioned into $r$
equal-sized sets (or subtables) of size $n/r$, and edges are chosen at random subject to the constraint that each edge contains exactly one vertex from each set. Second, 
the peeling process in our GPU implementation does not attempt to peel all vertices in each round. 
Instead, our GPU implementation proceeds in \emph{subrounds}, where each round consists of $r$ subrounds. In the $i$th subround of a given 
round, we remove all
the vertices of degree less than $k$ in the $i$th subtable. Note that running
one round of this algorithm is not equivalent to running one round of the
original parallel peeling algorithm. This is because peeling the first
subtable may free up new peelable vertices in the second subtable, and
so on. Hence, running one round of the algorithm used in our GPU implementation may remove more
vertices than running one round of the original algorithm.


In this section, we analyze the peeling process used in our GPU implementation.
We can use a similar approach as above to obtain the recursion for the
survival probabilities for this algorithm. Let $\rho_{i,j}$ be the
probability that a vertex in the tree survives $i$ rounds when it's in
the $j$th subtable, with each $\rho_{0, j} = 1$. Then,
\begin{align*}
\rho_{i,j} &= \Pr\bigg(\text{Poisson}\Big(rc \prod_{h < j}\rho_{i,h}\prod_{h > j}\rho_{i-1,h}\Big) \geq k-1\bigg).
\end{align*}
By the same reasoning,
\begin{align}
\label{eq:newlambdaij}
\lambda_{i,j} &= \Pr\bigg(\text{Poisson}\Big(rc \prod_{h < j}\rho_{i,h}\prod_{h > j}\rho_{i-1,h}\Big) \geq k \bigg)
\end{align}
where $\lambda_{0,j}=1$ for all $j$. Also, we can consider 
\begin{align*}
\beta_{i,j} &= rc \bigg(\prod_{h < j}\rho_{i,h}\bigg) \bigg(\prod_{h > j}\rho_{i-1,h}\bigg).
\end{align*}
These equations differ from our original equation in a way similar to how
the equations for standard multiple-choice load-balancing differ from 
V{\"o}cking's asymmetric variation of multiple-choice load-balancing, where
a hash table is similarly split into $r$ subtables, each item is given one
choice by hashing in each subtable, and the item is placed in the least loaded subtable,
breaking ties according to some fixed ordering of the subtables \cite{MVo,vocking}.  

Motivated by this, we can show that in this variation, below the threshold, these values eventually decrease
``Fibonacci exponentially'', that is, with the exponent falling according to a generalized
Fibonacci sequence.  
We follow the same approach as outlined in Section~\ref{sec:high}.
Let $\beta'_m = \beta_{i,j}$ where $m = (i-1)r+j$, and similarly for 
$\lambda'_m$ and $\rho'_m$, so we may work in a single dimension. 
Let $F_{r-1}(i)$ represent the $i$th number in a Fibonacci sequence of order $r-1$.  Here, a Fibonacci sequence of order $r$ is defined such that the first $r-1$ elements in the sequence equal one, and for $i > r-1$, the $i$th 
element is defined to be the sum of the preceding $r-1$ terms.

We choose a constant $I$ so that $\beta'_{I+a} \leq \phi^{F_{r-1}(a)}$ for an appropriate constant $\phi < 1$ and 
$0 \leq a \leq r-1$.  
We inductively show that
$$\beta'_{I+t} \leq \phi^{(k-1)^{\lfloor t/r \rfloor} F_{r-1}(t)}$$ when 
$\frac{rc}{[(k-1)!]^{r-1}} < 1$;  as in Section \ref{sec:belowthresh}, the proof can be modified easily if $\frac{rc}{[(k-1)!]^{r-1}} > 1$
by simply choosing a different (constant) starting point $I$ for the induction.  
In this case, for $t \geq r$
\begin{align}
\notag
\beta'_{I+t} &\leq \bigg[\prod_{I+t-r < j < I+t} \frac{(\beta'_j)^{k-1}}{(k-1)!}\bigg] rc \\
\notag
 &\leq \frac{rc}{[(k-1)!]^{r-1}} \prod_{I+t-r < j < I+t} {(\beta'_j)^{k-1}}   \\
 \notag
 &\leq  \frac{rc}{[(k-1)!]^{r-1}}  \prod_{I+t-r < j < I+t} \left ( {\phi^{F_{r-1}(j)(k-1)^{\lfloor (t-r)/r \rfloor}}}\right )^{(k-1)} \\
 \label{ineq:final}
 &\leq \phi^{(k-1)^{\lfloor t/r \rfloor} F_{r-1}(t)}.
\end{align}

Thus, our induction yields that the exponent of $\phi$ in the $\beta'_{m}$ values falls according to a generalized
Fibonacci sequence of order $r-1$, leading to an asymptotic constant factor reduction in the number of overall rounds,
even as we have to work over a larger number of subrounds.  
Inequality \eqref{ineq:final} applies to the idealized branching process,
but we can handle deviations between the idealized process and the actual process essentially as in Theorem \ref{thm:loglog}.
This yields the following variation of
Theorem \ref{thm:loglog} for the setting of peeling with sub-tables.

\begin{theorem}
\label{thm:fibonacci}
Let $r \geq 3$ and $k \geq 2$. Let $\phi_{r-1} = \lim_{k \rightarrow \infty} F_{r-1}^{1/k}(k)$ be the asymptotic growth rate for the Fibonacci sequence
of order $r-1$. Let $G$ be a hypergraph over $n$ nodes with $cn$ edges generated according to the following random process.
The vertices of $G$ are partitioned into $r$ subsets of equal size, and the edges are generated at random subject to the constraint that each edge contains exactly one vertex from each set. 

 With probability $1-o(1)$, the peeling process for the $k$-core in $G$
that uses $r$ subrounds in each round 
terminates after $\frac{1}{r \log \phi_{r-1} + \log (k-1)}\log \log n + O(1)$ rounds when $c < c^*_{k, r}$.
\end{theorem}


It is worth performing a careful comparison of Theorems \ref{thm:loglog} and 
\ref{thm:fibonacci}. For simplicity, we will restrict the discussion to $k=2$.
This corresponds to the case where we are interested in the 2-core of the hypergraph, as in our IBLT implementation.
Theorem \ref{thm:loglog} guarantees that 
the peeling process of Section \ref{sec:belowthresh} requires $\frac{1}{\log(r-1)} \log \log n + O(1)$. 
Meanwhile, Theorem \ref{thm:fibonacci} guarantees
that the total number of sub-rounds required by our IBLT implementation is $r \cdot \frac{1}{r \log \phi_{r-1}}\log \log n + O(1) = \frac{1}{\log \phi_{r-1}} \log \log n + O(1)$. Thus, parallel peeling with subtables
takes a factor $\log(r-1)/\log(\phi_{r-1})$ more (sub)-rounds
than parallel peeling without subtables. 

For $r=3$, $\phi_{r-1} \approx 1.61$ is the golden ratio, and in
this case $\log(r-1)/\log(\phi_{r-1}) \approx 1.456$. Thus, for $r=3$ and $k=2$, parallel peeling with sub-tables takes a factor of less than 1.5 times more (sub)-rounds than parallel peeling. In contrast,
one might a priori have expected that the number of 
sub-rounds for peeling with sub-tables would be a factor $r=3$ larger than in the standard peeling process, since $r$ serial steps are required to iterate through all $r$ subtables.

As $r$ grows,
$\phi_{r-1}$ rapidly approaches 2 from below. For example, for $r=4$ this quantity is approximately 1.83 and for $r=5$ it is approximately 1.92 \cite{vocking}.
It follows that for large $r$  the ratio $\log(r-1)/\log(\phi_{r-1})$ is 
very close to $\log_2(r-1)$. 

\subsection*{Simulations with Subtables}
\label{sec:experimentsubtable}

We ran simulations for the parallel peeling algorithm with subtables in a similar way as the simulations in 
Section \ref{sec:experiment}. Table \ref{table:simulationsubtable1} shows the results for the average number
of subrounds. The number of subrounds is at most $r$ times the number of rounds in the original 
parallel peeling algorithm,  but our analysis of Section~\ref{sec:peelsubtable} suggests the number of 
subrounds should be  significantly smaller.  In this case, comparing Table \ref{table:simulationsubtable1} with Table \ref{table:simulation1}, this factor is about 2.

\begin{table}
\small
\centering
\begin{tabular}{|l|r|r|r|r|}
\hline
& \multicolumn{2}{|c|}{$c = 0.7$} & \multicolumn{2}{|c|}{$c = 0.75$} \\
\hline
$n$ & Failed & Subrounds & Failed & Subrounds\\
\hline
10000 & 0 & 26.018 & 0 & 47.732 \\
20000 & 0 & 26.142 & 0 & 47.659 \\
40000 & 0 & 26.273 & 0 & 47.666 \\
80000 & 0 & 26.452 & 0 & 47.783 \\
160000 & 0 & 26.585 & 0 & 47.769 \\
320000 & 0 & 26.790 & 0 & 47.925 \\
640000 & 0 & 26.957 & 0 & 48.070 \\
1280000 & 0 & 27.006 & 0 & 48.141 \\
2560000 & 0 & 27.012 & 0 & 48.175  \\
\hline
\end{tabular}
\caption{Results of simulations of peeling with subtables using $r = 4$ and $k=2$, over $1000$ trials.}
\label{table:simulationsubtable1}
\end{table}

\begin{table}
\footnotesize
\centering
\begin{tabular} {|l|r|r|r|}
\hline
\multicolumn{4}{|c|}{$c = 0.7$} \\
\hline
$i$ & $j$ & Prediction & Experiment \\
\hline
1 & 1 & 942230 & 942230 \\
1 & 2 & 876807 & 876803 \\
1 & 3 & 801855 & 801855 \\
1 & 4 & 714875 & 714878 \\
2 & 1 & 678767 & 678771 \\
2 & 2 & 643070 & 643080 \\
2 & 3 & 609686 & 609697 \\
2 & 4 & 581912 & 581919 \\
3 & 1 & 554402 & 554414 \\
3 & 2 & 527335 & 527341 \\
3 & 3 & 500469 & 500476 \\
3 & 4 & 472470 & 472475 \\
4 & 1 & 442874 & 442871 \\
4 & 2 & 410958 & 410956 \\
4 & 3 & 375770 & 375764 \\
4 & 4 & 336458 & 336447 \\
5 & 1 & 292159 & 292144 \\
5 & 2 & 242396 & 242374 \\
5 & 3 & 187891 & 187866 \\
5 & 4 & 131789 & 131776 \\
6 & 1 & 80372 & 80376 \\
6 & 2 & 40582 & 40600 \\
6 & 3 & 15481 & 15503 \\
6 & 4 & 3649 & 3666 \\
7 & 1 & 348 & 354 \\
7 & 2 & 6 & 6 \\
7 & 3 & 0.003 & 0.008 \\
7 & 4 & 0 & 0 \\
\hline
\end{tabular}
\caption{\small Results of simulations of peeling with subtables showing how well the recursion for $\lambda'_{i,j}$ approximates the number of vertices left after $t$ rounds. The experiments are run using $r=4, k=2, n = 1$ million, averaged over $1000$ trials.}
\label{table:simulationsubtable2}
\end{table}

We also performed  simulations to determine how closely the recursion given in Equation~\eqref{eq:newlambdaij} predicts the number of vertices left after peeling the $j$th subtable in the $i$th round. 
Denote by $\lambda'_{i,j}$ the expected fraction of vertices left in the $(i, j)$'th subround. Then $\lambda'_{i,j}$ is given by the following formula:
\begin{align*}
\lambda'_{i,j} = \frac{1}{r}\bigg(\sum_{h \leq j}\lambda_{i,h} + \sum_{h > j}\lambda_{i-1,h}\bigg),
\end{align*}
where the $\lambda_{i,j}$ values are given by Equation \eqref{eq:newlambdaij}. The results are presented in Table \ref{table:simulationsubtable2}, where the prediction column reports the values of $\lambda'_{i,j}n$. As can be seen, the prediction closely matches the number of vertices left in the simulation.

\section{Proof of Theorem \ref{thm:distance}}
\label{app:distance}
We recall the statement of Theorem \ref{thm:distance}, before
offering a proof.
\medskip

\noindent \textbf{Theorem \ref{thm:distance}.} \emph{ Let $\nu = |c_{k, r}^*-c|$ for constant $c$ with $c < c_{k,r}$. With probability $1-o(1)$, peeling in $G_{n, cn}^r$ requires $\Theta(\sqrt{1/\nu}) + \frac{1}{\log((k-1)(r-1))}\log \log n$ rounds when $c$ is below the threshold density $c_{k, r}^*$. }

\medskip

Since $k$ and $r$ are constants, for notational convenience, we use $c^*$ in place of $c_{k,r}^*$ where the meaning is clear.
Recall that we are working in the setting where $\nu = c^*-c > 0$. Recall Equation \eqref{eq:cstar} for $c^*$ and let $x^*$ be the value of $x$ that satisfies $c^* = \frac{x}{r(1-e^{-x}\sum_{j=0}^{k-2}\frac{x^j}{j!})^{r-1}}$. 
Intuitively, one may think of $x^*$ as the expected number of surviving
descendant edges of each node in the graph 
when the edge density $c$ is precisely equal to the threshold density
$c^*$.

The heart of our analysis lies in proving the following lemma.
\begin{lemma}
\label{bigasslemma}
Let $\tau < x^*$ be any constant. It takes $\Theta(\sqrt{1/\nu})$ rounds before $\beta_i < \tau$.
\end{lemma}
\begin{proof}
Recall Equation \eqref{eq:betairecursion};  setting $\delta_i = \beta_i - x^*$ gives
\begin{align}
\label{eq:figureequation}
\beta_{i+1} &=\!\! \bigg[1 - e^{-\beta_i}\sum_{j=0}^{k-2}\frac{{\beta_i}^j}{j!}\bigg]^{r-1}rc \\
\notag
&=\!\! \bigg[1 - e^{-x^*-\delta_i}\sum_{j=0}^{k-2}\frac{(x^*+\delta_i)^j}{j!}\bigg]^{r-1}rc^*\!\!-\!\!\bigg[1 - e^{-x^*-\delta_i}\sum_{j=0}^{k-2}\frac{(x^*+\delta_i)^j}{j!}\bigg]^{r-1}r\nu \\
\notag
&=\!\! f(\delta_i) - g(\delta_i)\nu,
\end{align}

where $$f(\delta_i) = (1 - e^{-x^*-\delta_i}S(k-2,x^*+\delta_i))^{r-1}rc^*$$
and 

$$g(\delta_i) = \bigg[1 - e^{-x^*-\delta_i}\sum_{j=0}^{k-2}\frac{(x^*+\delta_i)^j}{j!}\bigg]^{r-1}r.$$
 
Then, using the Taylor series expansion for $f(\delta_i)$ around 0,
\begin{align}
\label{eq:deltaiappeq}
f(\delta_i) &= f(0) + f'(0)\delta_i + \frac{f''(0)}{2}\delta_i^2 + O(\delta_i^3) 
\end{align}

We claim that the right hand side of Equation \eqref{eq:deltaiappeq}
in fact equals 
\begin{align}
\label{eq:themajoreq}
& x^*+\delta_i-c_1\delta_i^2 + O(\delta_i^3),
\end{align}
for some constant  $c_1 > 0$.
In order to show this, we must prove three statements: First, that $f(0)=x^*$. Second, that $f'(0)=1$. Third, that $f''(0) = - c_1 < 0$.  
The first statement holds by definition of $x^*$. We now turn to proving the second statement.

\paragraph{Proof that $f'(0) =1$}
\noindent For convenience, in what follows, let $S(a,z) = \sum_{j=0}^{a} \frac{z^j}{j!}$, and note that $\frac{\text{d}S(a,z)}{\text{d}z} = S(a-1,z)$.  (For the case where $a=0$, we interpret $S(-1,z) = 0$.)  

To begin, recall that Equation \eqref{eq:cstar} expresses $c^*$ as
$\min_{x> 0} F(x)$, where \begin{equation*}
\label{eq:F} F(x) = \frac{x}{r\left(1-e^{-x}S(k-2,x)\right)^{r-1}},
\end{equation*}
and that $x^*$ is the value of $x$ that achieves the minimum.
Since $x^*$ is a local minimum of $F$, it must hold that
$F'(x^*)=0$. To ease calculations, let $G(x^*)=F(x^*)/r$:
since $F'(x^*)=0$, it holds that $G'(x^*)=0$ as well.  
Explicitly computing $G'(x^*)$, we see that:
\begin{eqnarray*}(1-e^{-x^*}S(k-2,x^*))^{1-r} - x^*(r-1)(1-e^{-x}S(k-2,x^*))^{-r}\\
\cdot
(e^{-x^*}S(k-2,x^*)-e^{-x^*}S(k-3,x^*)) = 0.\end{eqnarray*}

Standard manipulations then reveal:
\begin{eqnarray} 
\label{eq:starderiv}
e^{-x^*}(S(k-2,x^*)-S(k-3,x^*)) & = & \frac{1-e^{-x^*}S(k-2,x^*)}{x^*(r-1)}.
\end{eqnarray} 

Now recall that
$$f(\delta_i) = (1 - e^{-x^*-\delta_i}S(k-2,x^*+\delta_i))^{r-1}rc^*.$$

It follows that
\begin{eqnarray} 
\notag
f'(0) \\
\notag
& \!\!\!\!\!\!\!\!\!\!\!\!\!\!\!\!\!\!\!\!\!\!\!\!\!\!\!\!\!\!\!\!\!\!\!\!\!\!\!\!\!\!\!\!\!\!\!\!\!\!\!\! = &\!\!\!\!\!\!\!\!\!\!\!\!\!\!\!\!\!\!\!\!\!\!\!\!\!\!\! (r-1)rc^* (1 - e^{-x^*}S(k-2,x^*))^{r-2}e^{-x^*}(S(k-2,x^*)-S(k-3,x^*)) \\
\label{secondline}
      &\!\!\!\!\!\!\!\!\!\!\!\!\!\!\!\!\!\!\!\!\!\!\!\!\!\!\!\!\!\!\!\!\!\!\!\!\!\!\!\!\!\!\!\!\!\!\!\!\!\!\!\! = &\!\!\!\!\!\!\!\!\!\!\!\!\!\!\!\!\!\!\!\!\!\!\!\!\!\!\!\frac{rc^*}{x^*} (1 - e^{-x^*}S(k-2,x^*))^{r-1} \\
      \label{thirdline}
      &\!\!\!\!\!\!\!\!\!\!\!\!\!\!\!\!\!\!\!\!\!\!\!\!\!\!\!\!\!\!\!\!\!\!\!\!\!\!\!\!\!\!\!\!\!\!\!\!\!\!\!\!  = &\!\!\!\!\!\!\!\!\!\!\!\!\!\!\!\!\!\!\!\!\!\!\!\!\!\!\!\!\!\!\!\! 1.
\end{eqnarray} 
Here Equation \eqref{secondline} follows from Equation~\eqref{eq:starderiv}, and Equation \eqref{thirdline} follows from the definition
of $c^*$ and $x^*$ according to Equation \eqref{eq:cstar}.

\paragraph{Proof that $f''(0) < 0$}
\noindent After some tedious but straightforward calculations, we find  that
\begin{eqnarray} 
f''(0) &  = & \frac{r-2}{(r-1)x^*} - 1 + \frac{k-2}{x^*}.
\end{eqnarray} 
We therefore have that $f''(0) < 0$ as long as
\begin{eqnarray} \label{xstareq} 
x^* > k - 1 -\frac{1}{r-1}.
\end{eqnarray} 

Our argument will proceed as follows.
Equation \eqref{eq:cstar} implies that $x^*$ is a 
local minimum of the function $Z(x) = \frac{x}{(1-e^{-x}S(k-2,x))^{r-1}}$.
We will compute $Z'(x)$, and show that for $Z'(x) < 0$ for all $x \in (0, k-1)$ for any $r \geq 3$.  It will follow that $x^* \geq k-1$, and
hence Inequality \eqref{xstareq} holds. Details follow. 

It suffices to consider the function $rZ(x) = x(1-e^{-x}S(k-2,x))^{1-r}$,
as the derivative of $rZ(x)$ always has the sign as $Z(x)$. The derivative of $rZ(x)$ is    
\begin{align} \notag \left(1-e^{-x}S(k-2,x)\right)^{1-r} +\\ 
\notag x(1-r)\left(1-e^{-x}S(k-2,x)\right)^{-r}
 \cdot e^{-x}\left(S(k-2,x)-S(k-3,x)\right) = \\ 
  \left(1-e^{-x}S(k-2,x)\right)^{-r} \notag \\ \label{finalequationinpaper}
\cdot \left[\left(1-e^{-x}S(k-2,x)\right) + x^{k-1}e^{-x}(1-r)/\left(\left(k-2\right)!\right)\right].
\end{align}
We will show this the above expression is negative for all $x \in (0,k-1)$.  Note that 
$1-e^{-x}S(k-2,x) = e^{-x} \sum_{j=k-1}^\infty x^j/j! > 0$.  Hence, multiplying Expression \eqref{finalequationinpaper} through by
$(1-e^{-x}S(k-2,x))^{r}e^{x}$, we find the derivative is negative 
when 
$$\frac{(r-1)x^{k-1}}{(k-2)!} >  \sum_{j=k-1}^\infty x^j/j!.$$
Notice that the left hand side is $(r-1)(k-1) \geq 2(k-1)$ times the first term of the right hand side,
and for $x < k -1$, the terms in the summation on the right hand side are decreasing.  In fact,
after $k-1$ terms, the sum on the right hand side is dominated by a geometric series in which each term decreases by a factor of $1/2$.  It follows that right hand sum is less than $2(k-1)$ times the first term, and hence
the derivative is negative for all $x \in (0, k-1)$. This completes the proof that  $f''(0) < 0$, and we conclude that Equation \eqref{eq:themajoreq} 
holds.

\medskip
\medskip
\noindent  Equation \eqref{eq:themajoreq} combined with Taylor's Theorem implies that there exists some $h(\delta_i)$ such that $f(\delta_i) = x^*+\delta_i-c_1\delta_i^2 + h(\delta_i)\delta_i^2$, where $\lim_{\delta_i \to 0} h(\delta_i) = 0$. This means there exist constants $c_1', c_1''>0$ such that $x^* + \delta_i - c'_1\delta_i^2 <  f(\delta_i) < x^* + \delta_i - c''_1\delta_i^2$ for $|\delta_i|$ less than a suitably chosen small constant.

In the same way, we can find constants $c'_2, c''_2 > 0$ such that $c'_2 < g(\delta_i) < c''_2$ for $|\delta_i|$ less than a suitably small constant. Since $\beta_{i+1} = x^* + \delta_{i+1}$, we can examine the following recurrence for $\delta_{i+1}$:
\begin{align*}
\delta_{i+1} &= \delta_i - c_1\delta_i^2 - c_2\nu \\
\delta_{0} &= r(c^*-\nu) - x^*,
\end{align*}

where $c_1, c_2 > 0$.

Again, we can upper bound $\delta_0$ by a suitably small constant by taking $\nu$ small enough. Next, we show it takes $\Theta(\sqrt{1/\nu})$ rounds for $\delta_i < \tau-x^*$, proving the lemma.  (Note $\tau-x^* < 0$.)
We break the problem into three substeps: the number of rounds it takes to get from $\delta_0$ to $\Theta(\sqrt{\nu})$, from $\Theta(\sqrt{\nu})$ to $-\Theta(\sqrt{\nu})$, and from $-\Theta(\sqrt{\nu})$ to $\tau-x^*$.

\medskip
\noindent {\bf From $\Theta(\sqrt{\nu})$ to $-\Theta(\sqrt{\nu})$:}
Since $|\delta_i| = \Theta(\sqrt{\nu})$, from the recursion, $\Theta(\nu)$ is subtracted from $\delta_i$ in each round. Since this interval has length $\Theta(\sqrt{\nu})$, it takes $\frac{\Theta(\sqrt{\nu})}{\Theta(\nu)} = \Theta(\sqrt{1/\nu})$ rounds for this substep. 

\medskip
\noindent {\bf From $\delta_0$ to $\Theta(\sqrt{\nu})$:}
Since $\delta_i = \Omega(\sqrt{\nu})$, each round $\Omega(\nu)$ is subtracted from $\delta_i$. Intuitively, this means we may ignore the $-c_2\nu$ term, and the recursion becomes
\begin{align*}
\delta'_{i+1} = \delta'_i - c_1(\delta'_i)^2
\end{align*}
for a suitable constant $c_1 > 0$, with $\delta'_0 = \delta_0$. More formally, since $c_2, \nu > 0$, the sequence of $\delta'_i$ values require more rounds to reach $\Theta(\sqrt{\nu})$ than the sequence of $\delta_i$ values, so analyzing this recursion provides an upper bound on the number of rounds for $\lambda_i$ to fall from $\delta_0$ to $\Theta(\sqrt{\nu})$.

Let $\delta''_i = c_1\delta'_i$. Then the recursion can be rewritten as
\begin{align*}
\delta''_{i+1} = \delta''_i(1-\delta''_i).
\end{align*}
Let $\gamma_i = 1/\delta''_i$. Then $\gamma_{i+1} = \gamma_i + 1 + \frac{1}{\gamma_i-1}$, which implies $\gamma_i > \gamma_0 + i$. For any $\nu' > 0$, take $N$ such that $1/(N-1) < \nu'$. Then since $\gamma_i > i$, for all $i > N$, $\gamma_{i+1} < \gamma_i + 1 + \nu'$ and $\gamma_i < (1+\nu')i$ for sufficiently large $i$. Therefore, $\gamma_i = (1+o(1))i$ and $\delta''_i = \frac{1+o(1)}{i}$.
Thus, it takes $i = O(\sqrt{1/\nu})$ rounds for $\delta''_i$ (and hence $\delta_i$) to reach $\Theta(\sqrt{\nu})$.

\medskip
\noindent {\bf From $-\Theta(\sqrt{\nu})$ to $\tau-x^*$:}
By the same reasoning as the previous case, consider the recursion
\begin{align*}
\delta''_{i+1} = \delta''_i(1-\delta''_i).
\end{align*}
Consider the sequence backwards; the number of rounds from $-\Theta(\sqrt{\nu})$ to $\tau-x^*$ is equivalent to the number of rounds for the ``backwards'' recursion, starting from $\tau-x^*$ and going to $-\Theta(\sqrt{\nu})$.
The backwards recursion can be obtained by solving the quadratic equation for $\delta''_i$:
\begin{align*}
\delta''_i = \frac{1-\sqrt{1-4\delta''_{i+1}}}{2}.
\end{align*}
We can reverse the negative signs and look at the following recursion
\begin{align*}
\gamma_{i+1} &= \frac{\sqrt{1+4\gamma_i}-1}{2}; \\
\gamma_0 &= x^*-\tau.
\end{align*}
The Taylor series expansion for $\frac{\sqrt{1+4x}-1}{2}$ reveals that $\frac{\sqrt{1+4x}-1}{2}= x-x^2+O(x^3)$, and it can be shown that $\frac{\sqrt{1+4x}-1}{2} < x-\frac{1}{2}x^2$ for $0 < x < 2-\sqrt{2}$. It takes a constant number of steps to get from $\gamma_0 = x^*-\tau$ to $2-\sqrt{2}$, and then we can upper bound the number of steps needed by this recursion to reach $\Theta(\sqrt{\nu})$ by the recursion $\gamma'_{i+1} = \gamma'_i-\frac{1}{2}(\gamma'_i)^2$.  As with the previous case, it takes $O(\sqrt{1/\nu})$ rounds for $\gamma_i$ to reach $\Theta(\sqrt{\nu})$.
\end{proof}


We note that, again, the above analysis focuses on the idealized process, but we can handle deviations between the idealized process and the actual process essentially as in Theorem \ref{thm:loglog}.

Theorem~\ref{thm:distance} follows readily.  
Choose $\tau$ satisfying $$\tau < \big(\frac{rc^*}{[(k-1)!]^{r-1}}\big)^{-\frac{1}{(k-1)(r-1)-1}}$$ and $\tau < 1$. By Lemma \ref{bigasslemma}, it takes $\Theta(\sqrt{1/\nu})$ rounds before $\beta_i < \tau$. 
The argument in Section \ref{sec:high} shows that in the idealized branching process, $\beta_i$ drops off doubly exponentially in the number of rounds after that, giving the $\frac{1}{\log((k-1)(r-1))}\log \log n$ additive term. 
Finally, the argument in the proof of Theorem \ref{thm:loglog} 
shows that deviations from the idealized process result in $O(1)$
additional rounds with high probability.

Our three-phase analysis appears to accurately
capture the empirical evolution of the idealized recursion. For example, Figure~\ref{fig:appsimulation} shows the behavior of $\beta_i$ according to the idealized recurrence of Equation \eqref{eq:figureequation} for selected values of $c$ close to the threshold
when $k=2$ and $r=4$. In this case the threshold $c^*_{2, 4}$ is approximately $0.77228$, and we show the evolution of $\beta_i$ at $c=0.77$ and $c=.772$. The long ``stretch'' in the middle of the plots corresponds to the $\Theta(\sqrt{1/\nu})$ rounds required  during ``middle phase''  in our argument, in which $\beta_i$ falls from 
$(\Theta(\sqrt{1/\nu})$ to $-\Theta(\sqrt{1/\nu})$.  

\begin{figure*}
\begin{center} 
\includegraphics[width=2.5in]{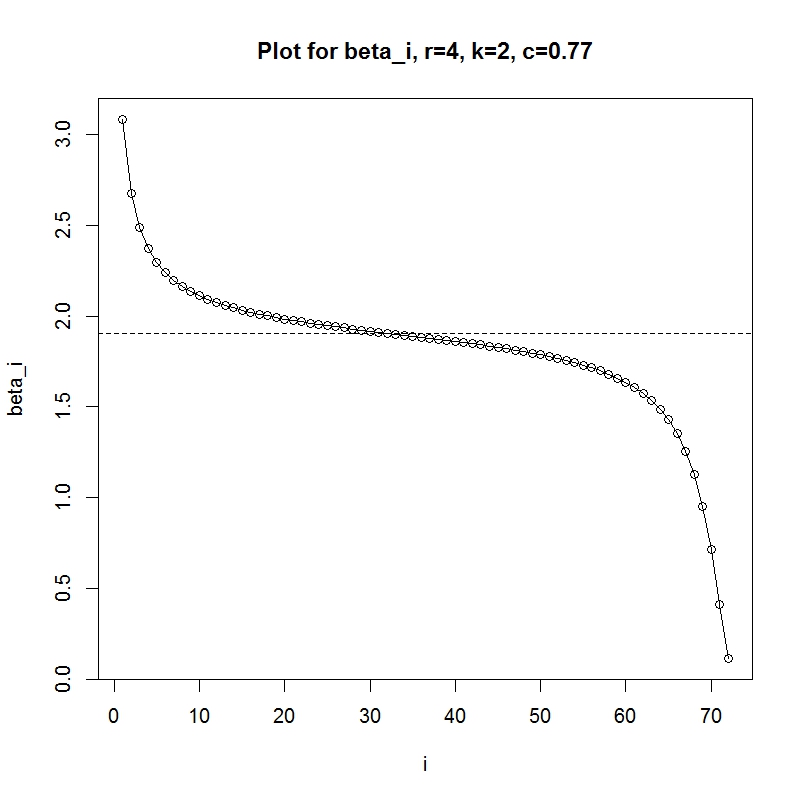} 
\mbox{                          }
\includegraphics[width=2.5in]{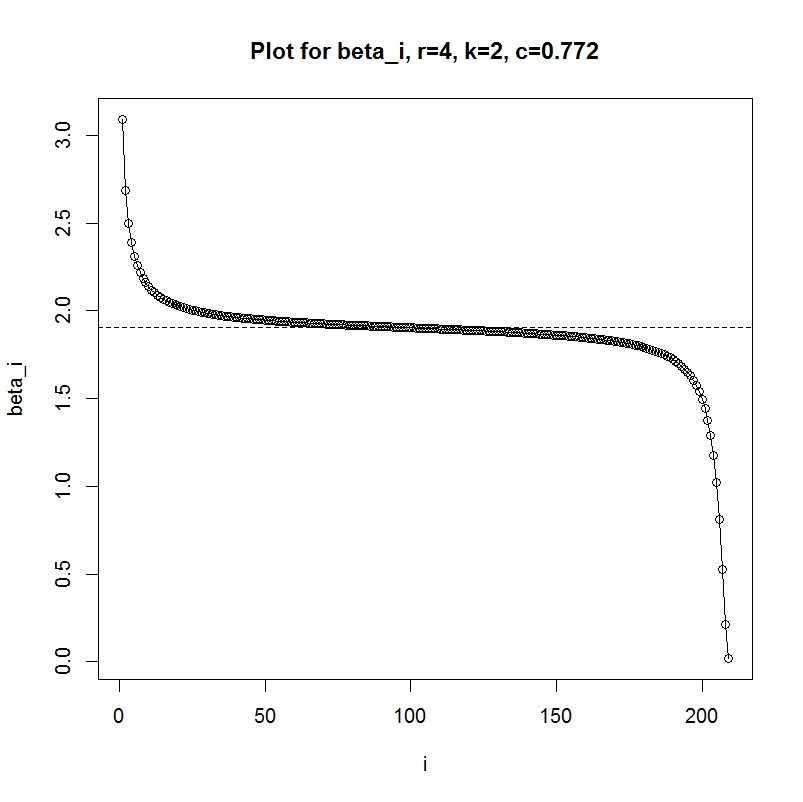}
\end{center} 
\caption{Behavior of the $\beta_i$ according to the idealized
recurrence of Equation \eqref{eq:figureequation} at values of $c$
close to the threshold density $c^*_{2, 4} \approx .77228$.}
\label{fig:appsimulation}
\end{figure*}

\end{document}